\documentclass{iucrjournals}

\usepackage[version=3]{mhchem} 

\graphicspath{{figures/}{\main/figures/}}
\usepackage{multicol}
\usepackage{graphicx}
\usepackage{svg}
\usepackage{amsmath}
\usepackage{amsthm}
\usepackage{amssymb}
\usepackage{sectsty}
\usepackage{fontenc}
\usepackage[labelfont=bf]{caption}
\usepackage{siunitx}
\newtheorem{theorem}{Theorem}
\newtheorem{corollary}{Corollary}[theorem]

\usepackage{graphicx}
\usepackage[labelformat=simple]{subcaption}
\usepackage[export]{adjustbox}
\usepackage{wrapfig}
\usepackage{float}
\usepackage{gensymb}
\usepackage{setspace}
\usepackage{xr-hyper}
\usepackage[utf8]{inputenc}

\title{A Crystallographic Metric for Continuous Quantification of Unit Cell Deformation}

\author[a]{Shannon Bernier\IUCrCemaillink{sbernier@jhu.edu}\IUCrOrcidlink{0000-0002-1036-6255}}
\author[a]{Gregory Bassen\IUCrEmaillink{gbassen1@jhu.edu}\IUCrOrcidlink{0009-0006-5683-7282}}%
\author[b]{Matthew Brem\IUCrEmaillink{mattbrem15@gmail.com}\IUCrOrcidlink{0009-0005-4505-2810}}%
\author[c]{Davor Tolj\IUCrEmaillink{dtolj1@jhu.edu}\IUCrOrcidlink{0000-0003-3806-1792}}
\author[d]{Quentin Simmons\IUCrEmaillink{nkt2pd@virginia.edu}\IUCrOrcidlink{0009-0002-3091-2222}}
\author[a,c,e]{Tyrel M. McQueen\IUCrCemaillink{mcqueen@jhu.edu}\IUCrOrcidlink{0000-0002-8493-4630}}

\affil[a]{Department of Chemistry, Johns Hopkins University, Baltimore, MD, United States}
\affil[b]{Independent researcher, Frederick, MD, United States}
\affil[c]{Institute for Quantum Matter, William H. Miller III Department of Physics and Astronomy, Johns Hopkins University, Baltimore, MD}
\affil[d]{Department of Physics, University of Virginia, Charlottesville, VA, United States}
\affil[e]{Department of Materials Science and Engineering, Johns Hopkins University, Baltimore, MD, United States}

\begin{document} 
\maketitle 

\begin{synopsis}
We present a simple and easy-to-use crystallographic metric, called the cubic deviation metric, which quantifies the distortion of a unit cell relative to a cubic geometry. The utility of the metric is illustrated with four case studies on pseudobrookites, quaternary homologous series, Wurtzite piezoelectrics, and cuprate superconductors.
\end{synopsis}

\begin{abstract}
Describing the deviation of a real structure from a hypothetical higher-symmetry ideal can be a powerful tool to understand and interpret phase transitions. Here we introduce a simple yet effective metric that quantifies the degree of unit cell distortion relative to a cube, called the cubic deviation metric. This enables continuous comparisons between unit cells of different geometries. We demonstrate the potential of this tool with four separate case study applications to real material systems: 1) discontinuous structural phase transitions in pseudobrookites; 2) homological structure classification; 3) structure-correlated piezoelectricity in hexagonal material; and 4) superconducting materials design in the cuprate family. Although this metric  does not replace detailed structural or group theory analysis, it enables comparison across different compositional and structural compound variants, even in the presence of disorder or absence of group–subgroup correlation. 
\end{abstract}


\section{Introduction}
Understanding solid-state phase transitions is crucial in materials science, as these transitions can significantly influence and modify material properties. Such transitions can be caused by temperature variation as small as a few degrees Celsius or the inclusion of less than a percentage point of a new dopant. In many cases, the change from one structure to another is equally subtle and within the error of laboratory diffraction, the choice of spacegroup is effectively arbitrary. This is a well-recognized problem in the materials science community. It is common to see a material referred to as "pseudocubic",\cite{Sinha2019Introduction/math,KuroiwaPiezoelectricityCations,Zaytseva2024ThreeAnions} "orthorhombically-distrorted"\cite{WangLocalAs}, "emergently tetragonal"\cite{Singh2024EmergentMaterial}, or similar during discussions of this phenomenon. Yet, such terms are imprecise and often carry different meanings for different authors or in different contexts. For example, \cite{Yazawa2021ReducedAl0.7Sc0.3N,Animitsa2009HydrationBa2CaWO6} refer to "tetragonal distortions" in an initially cubic ($a=b=c$) lattice experiencing uniaxial strain which produces lattice parameters $c\neq a=b$ indicative of a tetragonal lattice, while \cite{Singh2024EmergentMaterial} refers to "emergent tetragonality" as behavior indicative of a tetragonal system in an orthorhombic system ($a\neq b \neq c$) under strain but still possessing the lattice parameters of an orthorhombic phase. In a similar vein, the term "pseudocubic" is used to describe tetragonal or orthorhombic,\cite{Zaytseva2024ThreeAnions} monoclinic,\cite{Sinha2019Introduction/math}, and rhombohedral\cite{KuroiwaPiezoelectricityCations} materials.

The common cognition is that one spacegroup or unit cell may be very similar to another yet still distinct, and that in general it should be possible to say that one or another is closer to being "cubic", "tetragonal", or so forth. In the cases where the choice of unit cell is not so clear-cut - not coincidentally often the situations with the richest physics - this logic breaks down. The fundamental reason for this is that all possible unit cells may be represented as distinct "stops" on a continuum ranging from the highest possible symmetry, cubic, to the lowest possible triclinic Bravais lattice (Figure~\ref{fig:Bravais}). A tool for quantifying the stops on this continuum and the distance between them can provide a framework in which to concretely discuss the ambiguous terms mentioned above. The value of general continuous shape descriptors has been recognized for many years,\cite{Zabrodsky1992ContinuousMeasures} but new tools continue to be developed for application in analyses of unit cells, and other crystallographically-relevant shapes.\cite{Tuvi-Arad2024CSMAnalysis,Alon2023RESEARCHStructures,Mosca2020Voronoi-BasedLattices}

Here, we present one such tool which we refer to as the "cubic deviation metric" (CDM), as it classifies all unit cells on a continuum as distortions away from a  mathematically perfect cube. An earlier version of this metric was used in a previous paper by our group.\cite{Bernier2025Symmetry-mediatedPerovskite} Similar unit cell comparison tools exist, but these focus primarily on distinguishing the symmetry of two different spacegroups or sets of atomic positions, with or without a matching Bravais lattice.\cite{Mosca2020Voronoi-BasedLattices,ChisholmCOMPACK:Distances,DeLaFlor2016ComparisonServer} For the purposes of identifying equivalent descriptions of the same arrangement of atoms - a critical step especially as machine learning and artificial intelligence tools become capable of large batch structure prediction\cite{Merchant2023ScalingDiscovery} - these pre-existing tools are well-suited.
\begin{figure}[H]
  \centering
  \includegraphics[width=2.5in]{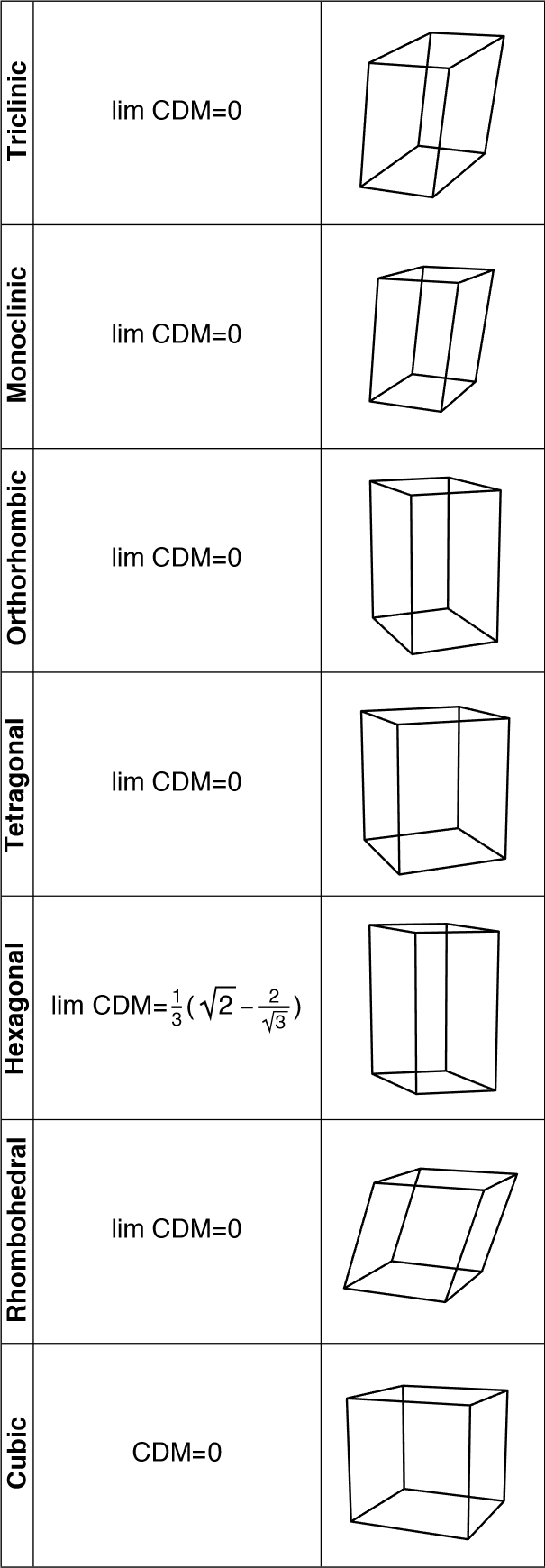}
  \caption{The 7 crystal shapes and their minimum possible cubic deviation metric values. All Bravais lattices may be represented as a mathematical object known as a parallelpiped with different lattices in the same Laue class having different atomic positions.}
  \label{fig:Bravais}
\end{figure}

Our metric, rather than supplanting these tools, aims to supplement by providing an easily-interpretable number describing a unit cell's fundamental shape. It is most similar to tools attempting to form a continuous relation between bonding polyhedra within unit cells such as the minimum bounding ellipsoid approach\cite{Cumby2017ARTICLEPolyhedra} or various algorithmic methods of relating bonding polyhedra to one or more "optimized" shapes\cite{Zabrodsky1992ContinuousMeasures,Alon2023RESEARCHStructures,Pinsky1998ContinuousPolyhedra,ALVAREZ2005ShapeChemistry,Alon2018ImprovedPermutations} by an identification of symmetry. Our metric for unit cells, similar to these tools for shape analysis \textit{inside} a unit cell, attempts to provide one single number describing materials and enable more quantitative comparison between them. While it has some similarities to the structure-specific tolerance factors,\cite{Goldschmidt1926,Mouta2013ToleranceStructures,Song2019ToleranceStructure,Song2020ToleranceStructure,Bassen2024ToleranceSeries,Tschauner2025TowardsPhases} used primarily to predict relationships before synthesis, it is rather designed as a structural analysis tool for use after synthesis, which can also be applied outside of a specific materials family. 

Further, as the materials discovery pipeline increasingly integrates generative machine learning for the prediction and design of candidate structures \cite{wilfong_ternary_2025}, there is a growing need for methods of down-selection to identify promising candidates. Here, we demonstrate a structure-property relationship, namely that the T{\text{c}} of cuprate superconductors exhibits a two-dome dependence on CDM, suggesting that this metric may serve as a valuable heuristic in the design of novel cuprate superconductors—and potentially other material classes as well. 

\section{Methods}
\vspace{-5mm}
\begin{figure}[ht]
\centering
\begin{subfigure}{0.55\textwidth}
  \includegraphics[width = \textwidth]{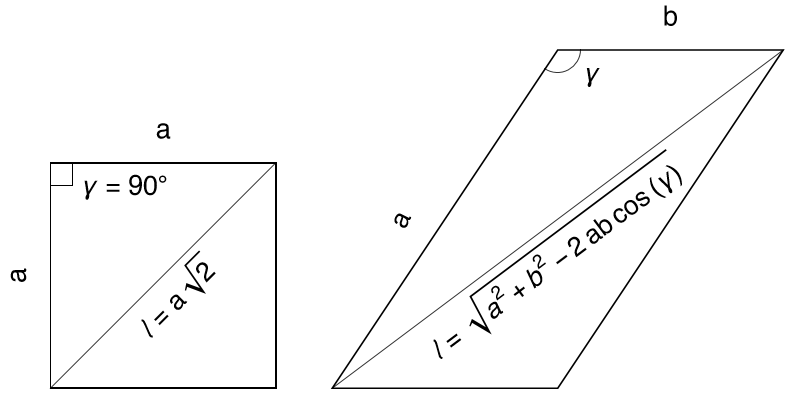}
  \caption{}
  \label{facediag}
\end{subfigure}

  \begin{subfigure}{0.55\textwidth}
    \includegraphics[width = \textwidth]{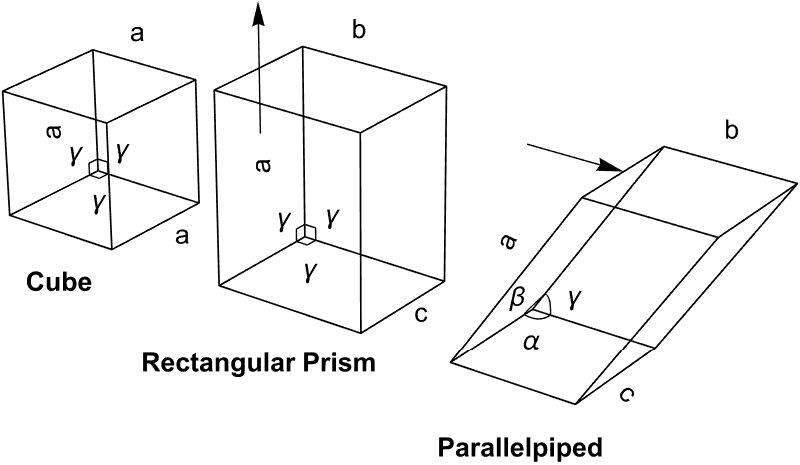}
    \caption{}
    \label{fig:3Dangles}
  \end{subfigure}
  \caption{a) Face diagonals of a parallelogram and a square as defined by Equation~\ref{lawofcosines}. The square is a special case of the more general shape where the face diagonal is equal to $a\sqrt{2}$. b) A series of generic parallelpipeds illustrating the definition of a, b, c, $\alpha$, $\beta$, and $\gamma$ in crystallography and Equation~\ref{eq:Mpoly}.}
  \label{bosscombofig}
\end{figure}
A parallelogram is a two-dimensional mathematical object (or polygon) with two pairs of parallel sides. The internal angles of these four-sided shapes can vary from \ang{0} to \ang{180}. A rectangle is a specific case of a parallelogram where all four angles are \ang{90}, and a square is a yet more specific case with equal angles and side lengths (Figure~\ref{facediag}). Similarly, the three-dimensional analogue of a parallelogram, the parallelpiped, is a generic class of polyhedron with three pairs of parallelograms for faces, Figure~\ref{fig:3Dangles} bottom left. When all internal angles are equal for all faces, the polyhedron is referred to as a right rectangular prism, and when the side lengths are as well, the shape becomes a cube. As will become important in a moment, every face of a cube is a square. Although some faces of other parallelpipeds may be squares, only a cube will have exclusively square faces.

All Bravais lattices are parallelpipeds with varying atom positions, as shown in Figure~\ref{fig:Bravais}. We are thus able to rely on well-established geometrical proofs for parallelpipeds to develop our "cubic deviation metric". The problem of distinguishing between a generic parallelpiped (P) and a cube may be reduced to checking if the three unique sides of P are squares (Theorem~\ref{thm:squares}). There are many properties of a parallelogram which may vary from a square, but not all will be sufficient to conclusively identify a square. For example, one can test for a rectangle by virtue of its unequal side lengths, but a rhombus (having equal side lengths but angles $\neq$ \ang{90}) will also pass this test. As will be shown below, all six lattice parameters will be required to distinguish cubic from non-cubic unit cells. We will additionally require other criteria to ensure maximum applicability to chemical systems.
\subsection{Metric requirements} \label{sec:MetReqs}
For a cubic deviation metric CDM which is both mathematically rigorous and chemically applicable, we identify the following six criteria:
\begin{enumerate}
    \item CDM must take into account all six lattice parameters - the three side lengths $a,b,c$ and three angles $\alpha,\beta,\gamma$ which describe the unit cell - and not include other chemical information (\textit{e.g.} atomic positions, radii, or stoichiometry).
    \item Since doubling a unit cell does not change its shape, CDM must be normalized such that two unit cells with the same ratio of angles and lengths give the same value. 
    \begin{itemize}
        \item This requirement has the added benefit of producing a metric which is unitless, enabling CDM to be readily compared between materials families.
    \end{itemize}
    \item CDM should not depend on the choice of crystallographic directions and should produce the same value if, \textit{e.g.} $a$ and $b$ are swapped for some given unit cell.   
    \item CDM must equal a specific value if and only if the unit cell is cubic. There must not be any cases where a non-cubic unit cell produces this value.
    \begin{itemize}
        \item Here, we have chosen CDM $= 0$ as the value for cubic unit cells, meaning values greater than 0 will indicate greater distortion from a perfect cube.
        \item For clarity only, we have also chosen to fix the maximum of CDM at 1.
    \end{itemize}
    \item Finally, CDM should display an obvious trend as a unit cell moves “further from a cube” by varying just one parameter, although a trend may not be obvious when varying multiple parameters at once (this is the problem our metric is being developed to solve).
\end{enumerate}

The simplest metric meeting these criteria arises through the comparison of face diagonals, which are defined by the Law of Cosines (Equation \ref{lawofcosines}). The inputs to the Law of Cosines are the lattice parameters only, satisfying requirement~1. We begin with a "square deviation metric" for a single face, $M_{face}$ (Figure~\ref{facediag}). In this simple 2D case, there is only one unique face and thus we need two terms for the two diagonals in this single face. In higher dimensionalities, there will be two terms for each of $N$ unique faces (Equation~\ref{eq:NormCondish}), so we will divide our final result by $N$ in order to satisfy both parts of requirement~4. In this trivial case, $N=1$ and the maximum value is unaffected. 

In all squares (and only squares, as shown in Theorem~\ref{thm:squares} in the Appendix) the ratio of side length to face diagonal is $\frac{a}{a\sqrt{2}}=\frac{1}{\sqrt{2}}$. By comparing ratios rather than lengths directly, we satisfy requirement 2 above. Subtracting the idealized value from the ratio and taking the absolute value produces a number which is larger when the deviation is larger. This allows us to distinguish between squares and parallelograms and satisfy requirement~5.
\begin{multicols}{2}
\begin{equation}
    l=\sqrt{a^2+b^2-2ab\cos{(\gamma)}}
    \label{lawofcosines}
\end{equation}
\begin{equation}
     N(D)=\frac{D!}{2(D-2)!}
     \label{eq:NormCondish}
\end{equation}
\end{multicols}

The difference of this ratio can uniquely identify squares. However, the formalism does not yet satisfy requirement~3 because the functional form $\left|\frac{1}{\sqrt{1-\cos\theta}}\right|$ is asymmetric, with a different behavior on either side of a cusp at $\theta=90\degree$. When $\theta<90\degree$, its value increases rapidly; when $\theta>90\degree$, its value increases much more slowly (see Figure~\ref{fig:piecewise}). There are two unique angles in each parallelogram, both of which are free to vary from \ang{0}-\ang{180} with the sum equal to \ang{180}. To address requirement~3, one can either restrict the inputs of the metric to only angles between \ang{90} and \ang{180}, requiring crystallographers to always select the larger of the two angles between a pair of sides when applying the metric, or equivalently define the function piecewise as in Equation~\ref{eq:PiecewiseFace}. The benefit of the piecewise definition is that one may refer to a fixed angle at all times and analyze the cubic/square deviation as a function of a single angle.
\begin{equation}
m=
\begin{cases}
    \sqrt{a^2+b^2-2ab\cos{(\gamma-180\degree)}} &  \gamma < 90\degree\\
    \sqrt{a^2+b^2-2ab\cos{(\gamma)}} &  \gamma  \geq 90\degree
\end{cases}
    \label{eq:PiecewiseFace}
\end{equation}
\begin{equation}
     \mathcal{M}_\text{face}= \frac{1}{N}\left(\left|\ \frac{a}{m}-\frac{1}{\sqrt{2}} \right|\ + \left|\ \frac{b}{m}-\frac{1}{\sqrt{2}} \right|\right)\ =\left|\ \frac{a}{m}-\frac{1}{\sqrt{2}} \right|\ + \left|\ \frac{b}{m}-\frac{1}{\sqrt{2}} \right|
     \label{Mface}
\end{equation}
We are thus able to derive Equation \ref{Mface} which has a term for each face diagonal in the shape and satisfies all five of our requirements. If the parallelogram is a square, ${M}_{face}=0$. Next, we extend this metric into three dimensions. The six face diagonals $l$ of a parallelpiped are equal across all faces only when every face is a square - in other words, when the parallelepiped is a cube. Equation~\ref{eq:NormCondish} gives the number of unique faces as three. Summing $\mathcal{M}_\text{face}$ over each unique face and dividing by $N=3$ produces the six-term metric $\mathcal{M}_\text{poly}$ for the three-dimensional parallelpiped shown in Figure~\ref{fig:3Dangles}, hereafter referred to as CDM for brevity.

The resulting "cubic deviation metric" CDM varies between 0 and 1, with a value of 0 indicating the shape is a cube. It relies only on lattice parameters, requirement 1. Rigorous mathematical proofs that this metric meets requirements 2-4 are presented in the Appendix. It should be noted that the metric is not linear with any one lattice parameter and that there may be multiple structures with the same value, including theoretically "maximally deviated" structures with value of CDM $=1$. However, it still follows an obvious trend from "less cubic" to "more cubic" when varying a single parameter, meeting requirement 5. The method of construction holds for objects of any dimensionality; for example, tesseracts (4D hypercubes) will require a total of 12 terms for face diagonals with $N=6$.

\begin{equation}\label{eq:Mpoly}
\begin{split}
  \mathcal{M}_\text{poly}&= \frac{1}{N}(\mathcal{M}_{\text{face}_{ab}}+ \mathcal{M}_{\text{face}_{ac}}+ \mathcal{M}_{\text{face}_{bc}})\\
&=\frac{1}{3}\begin{cases}
      \left|\ \frac{a}{\sqrt{a^2+b^2-2ab\cos{(180\degree-\gamma)}}}-\frac{1}{\sqrt{2}} \right|\ + \left|\  \frac{b}{\sqrt{a^2+b^2-2ab\cos{(180\degree-\gamma)}}}-\frac{1}{\sqrt{2}} \right|,\  & \gamma <90\degree\\
      \left|\ \frac{a}{\sqrt{a^2+b^2-2ab\cos{(\gamma)}}}-\frac{1}{\sqrt{2}} \right|\ + \left|\  \frac{b}{\sqrt{a^2+b^2-2ab\cos{(\gamma)}}}-\frac{1}{\sqrt{2}} \right|,\  & \gamma \geq90\degree\end{cases}
      \\ 
     &+ \frac{1}{3}\begin{cases}
      \left|\ \frac{c}{\sqrt{a^2+c^2-2ac\cos{(180\degree-\beta)}}}-\frac{1}{\sqrt{2}} \right|\ +
    \left|\ \frac{a}{\sqrt{a^2+c^2-2ac\cos{(180\degree-\beta)}}}-\frac{1}{\sqrt{2}} \right|,\ & \beta <90\degree\\ 
    \left|\ \frac{c}{\sqrt{a^2+c^2-2ac\cos{(\beta)}}}-\frac{1}{\sqrt{2}} \right|\ +
    \left|\ \frac{a}{\sqrt{a^2+c^2-2ac\cos{(\beta)}}}-\frac{1}{\sqrt{2}} \right|,\ &\beta \geq90\degree\\
    \end{cases}
    \\
   &+ \frac{1}{3}\begin{cases}
    \left|\ \frac{b}{\sqrt{b^2+c^2-2bc\cos{(180\degree-\alpha)}}}-\frac{1}{\sqrt{2}} \right|\ +
    \left|\ \frac{c}{\sqrt{b^2+c^2-2bc\cos{(180\degree-\alpha)}}}-\frac{1}{\sqrt{2}} \right|,\ & \alpha <90\degree\\
    \left|\ \frac{b}{\sqrt{b^2+c^2-2bc\cos{(\alpha)}}}-\frac{1}{\sqrt{2}} \right|\ +
    \left|\ \frac{c}{\sqrt{b^2+c^2-2bc\cos{(\alpha)}}}-\frac{1}{\sqrt{2}} \right|,\ & \alpha \geq90\degree\\
    \end{cases}
    \end{split}
\end{equation}

\subsection{Note on choice of unit cell and origin}
We now take a moment to discuss the limits of applicability of the CDM, which will be further explored in our case studies below. First, we comment on the hexagonal lattice, the crystallographic descriptor which looks most different from a cube in its most common description. As illustrated in Figure~\ref{fig:Bravais} and Figure~\ref{fig:HexRhomb}, although conventionally drawn as a hexagonal prism rather than a parallelpiped, the hexagonal unit cell is actually in the shape of a parallelpiped with the atomic positions possessing $C_3$ symmetry about the $c$ axis. When considering shape only, and not atomic positions, even this nominally disparate case can be placed in the middle of a continuum relating its shape to the cubic lattice. It is the atomic positions which make a hexagonal description useful for understanding the bonding in solids with these space groups, but as the CDM is agnostic to atomic position, there is nothing "special" about this particular lattice compared to the others. Similarly, other Bravais lattices in a crystal system produce equal CDM values when their lattice parameters are equal, regardless of centering. For example, face centered cubic, body centered cubic, and simple cubic lattices all have CDM=0 and hP and hR hegaonal lattices have equal nonzero values for the same lattice parameters.
\begin{figure}[ht]
\centering
  \includegraphics[width=0.3\textwidth]{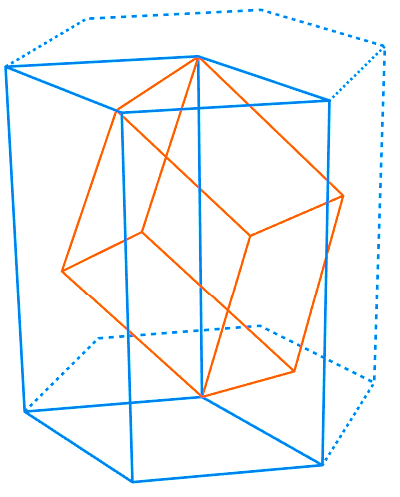}
  \caption{The location of the rhombohedral unit cell (orange) definition relative to the definition of the hexagonal unit cell (blue). The dashed blue region depicts the conventional representation of the hexagonal lattice.}
  \label{fig:HexRhomb}
\end{figure}

In the trigonal system, a hexagonal unit cell may be redefined as a rhombohedron with the directions of the crystal axes rotated as shown in Figure~\ref{fig:HexRhomb} using the relations in Equation~\ref{eq:HexToRhomb}. By design, CDM is agnostic to unit cell direction and thus may be applied to either setting of this lattice. In the hexagonal setting, the parallelpiped of interest has $\alpha=\beta=$ \ang{90}, $\gamma=$ \ang{120}, and $a = b$, with $c$ as a free parameter. CDM can be simplified for this lattice (Equation~\ref{eq:Mhex}) with a minimum possible value of $\approx$ 0.0809 when $c=a$, Equation~\ref{eq:Mhexmin}. There is no restriction on the values of $a$ and $b$ in this lattice, so although it is possible to reach this value  no amount of distortion can produce a lower value. Similar simplification of CDM is possible for the rhombohedral setting (Equation~\ref{eq:Mrhomb}). By definition, $\alpha,~\beta,~\gamma \neq 90\degree$ so the rhombohedral lattice's limit of CDM $=0$ (Equation~\ref{eq:Mrhombmin}) is never reached (without changing to the cubic lattice), but this description can produce smaller values of CDM than the hexagonal spacegroups.
\begin{equation}
    \lim_{c\to a}\text{CDM}_{\text{hexagonal}}=\frac{2}{3}\left|\ \frac{1}{\sqrt{3}}-\frac{1}{\sqrt{2}} \right|\ = \frac{1}{3}\left(\sqrt{2}-\frac{2}{\sqrt{3}}\right)\approx0.0809
    \label{eq:Mhexmin}
\end{equation}
\begin{equation}
    \lim_{\alpha\to 90\degree}\text{CDM}_{\text{rhombohedral}}=2\left|\ \frac{1}{\sqrt{2-2(0)}}-\frac{1}{\sqrt{2}} \right|\ = 0
    \label{eq:Mrhombmin}
\end{equation}

It is possible to convert from the hexagonal lattice to the rhombohedral, however, in the case where $a_{\text{h}}=c_{\text{h}}$ - the maximum possible "cubicity" of the hexagonal lattice - substituting the appropriate relations into the metric shows that $\alpha_{\text{r}}\approx97.2\degree$ producing CDM $\approx0.0865$. In other words, as one might naively expect, redefining a noncubic hexagonal lattice as rhombohedral does not produce a cubic unit cell. Somewhat unintuitively, however, there is a slight difference in the value of CDM for different descriptions of the same arrangement of atoms. The same is also true when describing a single unit cell versus a supercell with different symmetry or other choices where the symmetry operations change but the bonding and atomic positions remain constant. For this reason, care should be taken when comparing CDM values across different definitions of the unit cell; when the crystallographic directions defining $a,~b,~c$ change, there may be a discontinuity in the metric. This is not a failure of the metric, but rather a reminder that the unit cell formalism is an attempt to describe a high-symmetry arrangement of atoms in space and may not be a unique description.

To summarize, the minimum possible values of CDM for each crystal system are tabulated in Figure~\ref{fig:Bravais}. In real materials, the unit cell is selected based on observed atomic positions. Between two unit cells, equal lattice parameters will produce equal CDM values even when atomic positions vary significantly. In contrast, different axes choices when making this selection can affect the possible values of CDM. Axis and origin choices, supercells, and multiple coexisting unit cells (\textit{i.e.} multiple phases) should be considered when attempting to interpret differences in CDM between materials.

\section{Case studies}
CDM can be applied in a number of useful ways and is especially helpful when multiple lattice parameters are changing simultaneously. Here, we present four case studies of the metric applied to real materials data, highlighting its features and utility in materials problems.

\subsection{Metric behavior during phase transitions}\label{sec:1}
Pseudobrookite (Fe\textsubscript{2}TiO\textsubscript{5}) was first identified in 1878, and its crystal structure resolved in 1930.\cite{Koch1878,Pauling1930} The most extensively investigated orthorhombic pseudobrookites have compositions of the type M\textsuperscript{3+}\textsubscript{2}Ti\textsuperscript{4+}O\textsubscript{5} (M = Sc, Cr, Fe, Ti, Ga, Al) or M\textsuperscript{2+}Ti\textsuperscript{4+}\textsubscript{2}O\textsubscript{5} (M = Mg, Fe, Co).\cite{Tiedemann1982,Xirouchakis2007} In these structures, each iron-centered octahedron shares one edge with another iron-centered octahedron and three edges with titanium-centered octahedra. Conversely, each titanium-centered octahedron shares all six of its edges with iron-centered octahedra (Figure~\ref{fig:PBabc}a). The resulting network forms c-axis-oriented double chains of distorted octahedra, which are weakly bonded through shared edges. Unusual thermal expansion anisotropy in these systems and complex magnetic behavior has been strongly linked to structural properties.\cite{Bayer1971,Lang2019} 

Sustained interest in this class of compounds has generated an extensive body of literature and structural data. However, variability in sample preparation methods, thermal treatments, structural classifications (e.g. \textit{Cmcm}/\textit{Bbmm}), and structural refinement approaches has led to considerable discrepancies in reported lattice parameters, even among compounds with nominally identical compositions. Furthermore, standard conventions for labeling lattice parameters in different space groups can complicate comparisons between compounds. For instance, the shortest lattice parameter in the \textit{Cmcm} phase is unit cell length $a$, while the corresponding parameter in \textit{C2/m} is $b$. Among pseudobrookites, aluminum titanates (Al\textsubscript{2-x}Ti\textsubscript{1+x}O\textsubscript5) are the most widely applied and one of the best-characterized members of the family.\cite{Goldberg1968} These compounds can serve as ideal case studies for investigating the structural evolution using a cubic deviation metric, both as a function of temperature and composition.

\begin{figure}[ht]
\centering
\includegraphics[width=300pt]{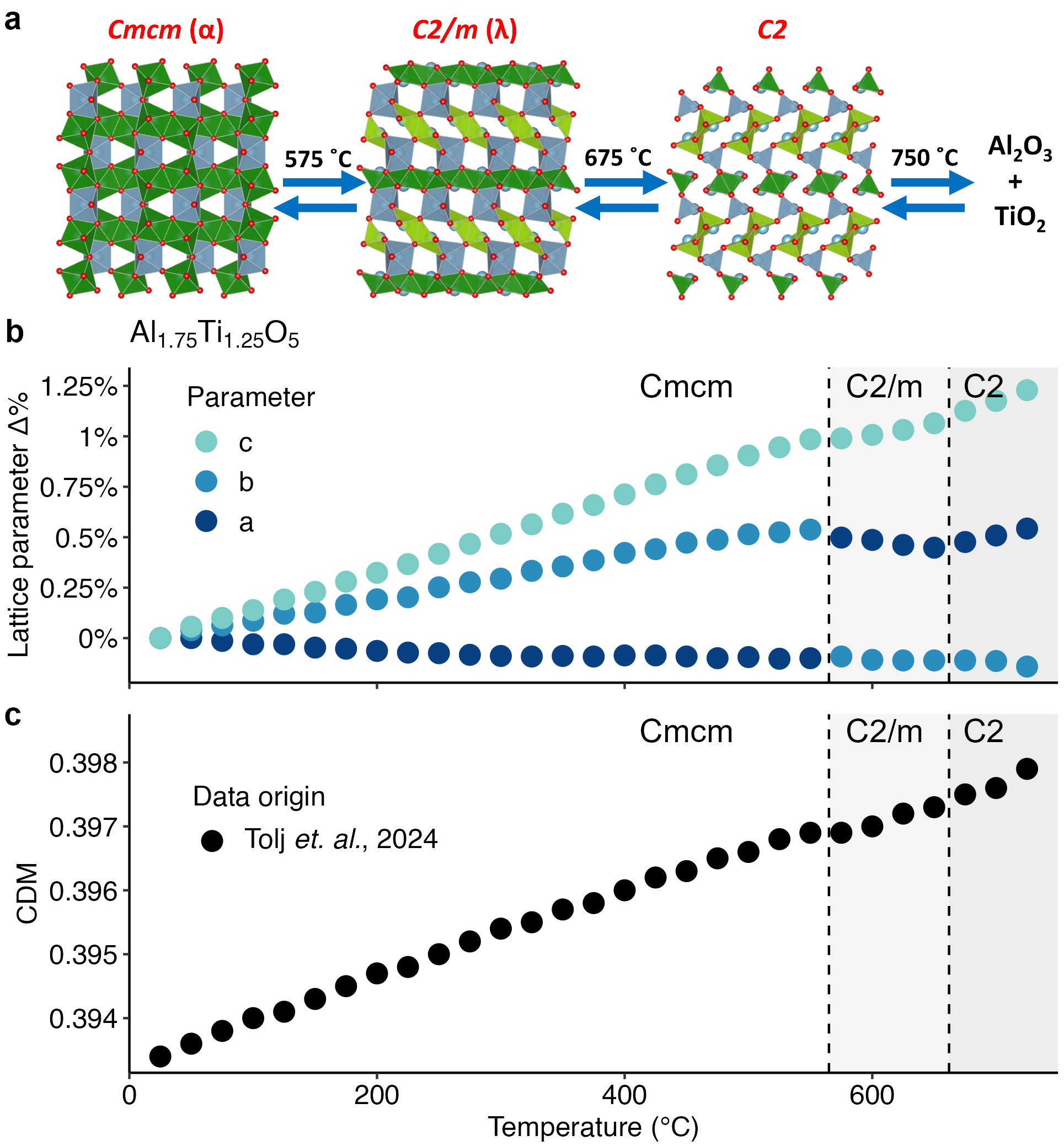}
\caption{a) 3D packing view of Al\textsubscript{1.75}Ti\textsubscript{1.25}O\textsubscript{5} structure at different temperatures. Temperature dependence of a change in: b) lattice parameters and c) cubic deviation metric. Data taken from \protect\cite{Tolj2024High-TemperaturePseudobrookite}.}
\label{fig:PBabc}
\end{figure}

In temperature driven structural evolution, exemplified by Al\textsubscript{1.75}Ti\textsubscript{1.25}O\textsubscript{5} (Figure~\ref{fig:PBabc}), the phase change from orthorombic \textit{Cmcm} to monoclinic \textit{C2/m} is associated with a gradual transformation in the coordination environment of one metal site from a distorted six-fold octahedron at room temperature to a five-fold square pyramidal geometry at 550 \degree C, Figure~\ref{fig:PBabc}a.\cite{Tolj2024High-TemperaturePseudobrookite} Additionally, there is a reduction in Al and Ti coordination numbers followed by the emergence of non-random cation distributions. Further increase in temperature induces additional structural distortion, leading to the symmetry reduction to the monoclinic \textit{C2} space group at 625 \degree C, and eventual thermal decomposition into Al\textsubscript{2}O\textsubscript{3} and TiO\textsubscript{2} above 750 \degree C. 

This transition may be followed by observing the differences in lattice parameters plotted in Figure~\ref{fig:PBabc}b, but it is complicated by the difference in definition of the $a$ and $b$ parameter between the \textit{Cmcm} and \textit{C2/m} spacegroups and anisotropic thermal expansion of the latice. Alternatively, we may apply CDM (Figure~\ref{fig:PBabc}c) and observe the general trend without this confusion. Both symmetry changes are second-order transitions and manifest as discontinuities in the CDM. With a numerical tool to observe these transitions, we can also compare the behavior of various pseudobrookite compositional variants as shown in Figure~\ref{fig:PBextradata} and observe a decrease in cubicity with increasing temperature in several structural analogues. More temperature data for these families could potentially reveal additional transitions.
\begin{figure}[ht]
\includegraphics[width=0.95\textwidth]{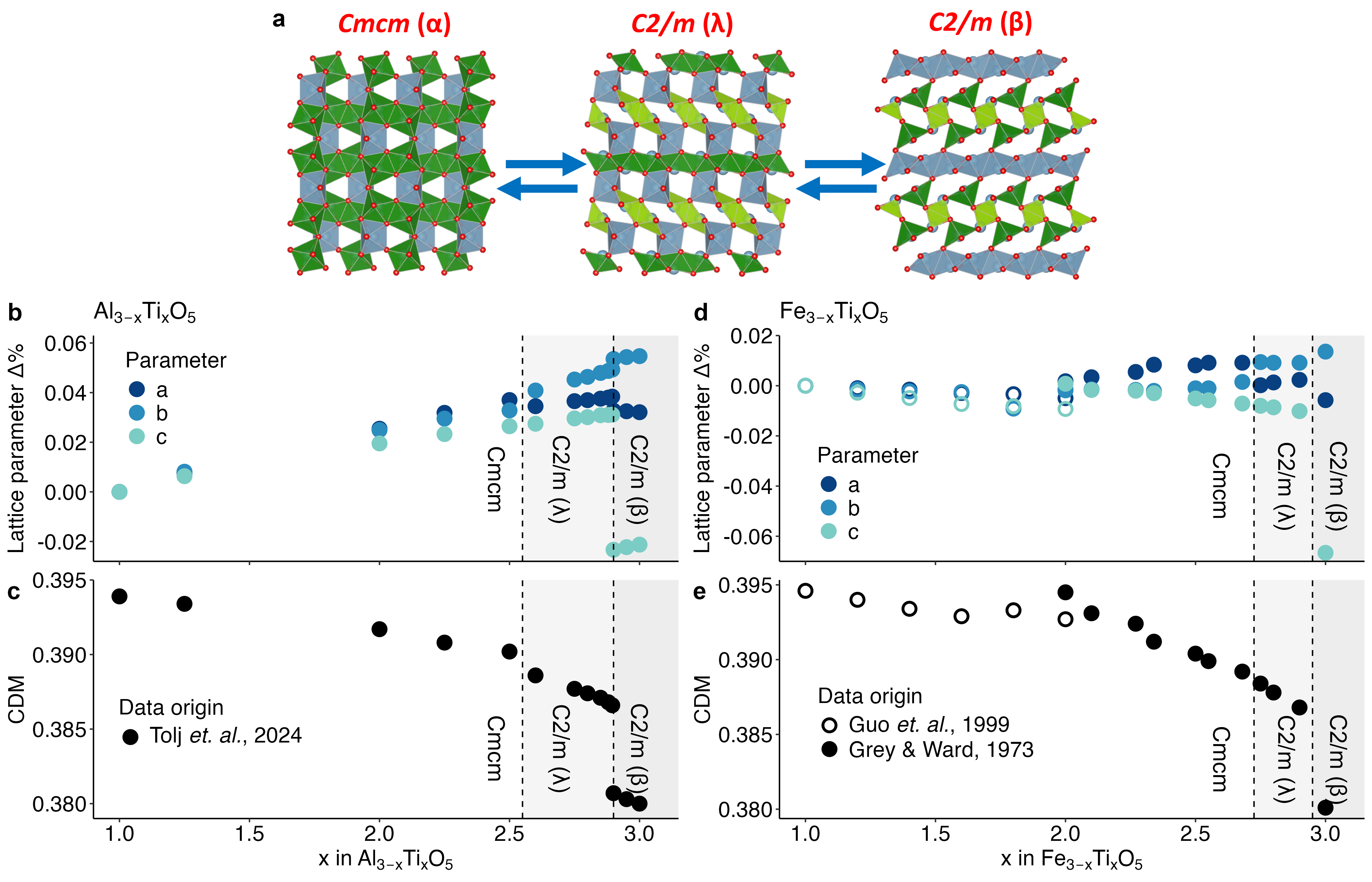}
\caption{M\textsubscript{3-x}Ti\textsubscript{x}O\textsubscript{5} composition range. a) 3D packing view of three phases present. Ti content dependence of a change in: Al\textsubscript{3-x}Ti\textsubscript{x}O\textsubscript{5} b) lattice parameters, c) cubic deviation metric; Fe\textsubscript{3-x}Ti\textsubscript{x}O\textsubscript{5} (open and full markers indicate two different data sources) d) lattice parameters, e) cubic deviation metric. Data from \protect\cite{Tolj2024High-TemperaturePseudobrookite,Guo1999,Grey1973}, as indicated in the figures.}
\label{fig:PBcomp}
\end{figure}

In contrast, composition-driven transitions in titanium-rich systems are governed by the formation and coupling of nonmagnetic spin-singlet Ti\textsuperscript{3+}-Ti\textsuperscript{3+} dimers with the lattice.\cite{Takahama2020} An initial increase in Ti content from Al\textsubscript{2}\textsuperscript{3+}Ti\textsuperscript{4+}O\textsubscript{5} leads to a paramagnetic, mixed-valent titanium state. Further increase in Ti content enhances dimer correlations, reducing the number of unpaired Ti\textsuperscript{3+} ions and driving a second-order phase transition, from orthorhombic \textit{Cmcm} ($\alpha$ phase) to a lower-symmetry monoclinic \textit{C2/m} ($\lambda$ phase). The charge ordering of Ti\textsuperscript{3+} dimers and Ti\textsuperscript{4+} in compounds close to Ti\textsubscript{3}O\textsubscript{5} leads to a first-order phase transition to \textit{C2/m} ($\beta$ phase) accompanied by abrupt changes in the lattice parameters and CDM.

We may again observe this transition by plotting CDM as in Figure~\ref{fig:PBcomp} c,e. Interestingly, CDM values decrease as we approach phase transitions with increasing Ti content. Additionally, we can also observe a sharp discontinuity in the data as the bonding changes from M\textsubscript{3-x}Ti\textsubscript{x}O\textsubscript{5} to Ti\textsubscript{3}O\textsubscript{5}, indicating first order transition. Both transitions occur at  slightly different points in the Al and Fe compounds due to the interaction between magnetic iron and titanium dimers, with the transition to the \textit{C2/m} ($\beta$ phase) occurring only after complete substitution to Ti\textsubscript{3}O\textsubscript{5}. CDM offers much clearer overview compared to an analysis of lattice parameters independently, especially in case of Fe\textsubscript{3-x}Ti\textsubscript{x}O\textsubscript{5} in which there are two compositional data sources, as no study has reported the full range data.\cite{Guo1999,Grey1973}

Despite the difference in origin, in both the temperature and compositional variation studies, the evolution of the structure shows a phase transition from the orthorhombic \textit{Cmcm} ($\alpha$ phase) to a lower-symmetry monoclinic \textit{C2/m} ($\lambda$ phase). On its own, this might suggest a common underlying mechanism. However, the cubic deviation metric CDM reveals distinct differences in behavior without the need for in-depth analysis presented here. Increasing temperature in Al\textsubscript{1.25}Ti\textsubscript{1.75}O\textsubscript{5} leads to an increase of CDM values, indicating distortion away from cubicity, while increase in Ti content in M\textsubscript{3-x}Ti\textsubscript{x}O\textsubscript{5} to Ti\textsubscript{3}O\textsubscript{5} leads to increase in lattice cubicty. This contrasting behavior highlights the utility of CDM, especially in systems with thermal expansion anisotropy (such as pseudobrookites) where it would not be possible to observe such trends simply from the change in lattice parameters. Although this analysis does not replace detailed structural or group theory analyses, it offers a simple yet effective tool to compare and classify phase transitions across large families of compounds. CDM allows for direct comparison across multiple data sources of nominally the same compound family but with significant variation in the reported lattice parameters due to synthesis or structural refinement methods. Notably, it also facilitates comparison across different compositional and structural variants, even in the presence of disorder or the absence of group–subgroup relationships. This singular parameter, based just on the lattice parameters, allows for a quick identification of interesting points in otherwise complex phase diagrams without the need for in-depth analysis.

\subsection{Comparing tolerance factors and the cubic deviation metric}\label{sec:2}
There is a subtle yet important distinction between what is known as a "tolerance factor" and a metric such as the one presented here. Tolerance factors - such as the Goldschmidt tolerance factor for perovskites \cite{Goldschmidt1926} and the Quaternary Tolerance Factor (QTF) developed for quaternary homologous chalcogenides \cite{Bassen2024ToleranceSeries}- are generally designed to be used before synthesis as a predictor of stability or structure type through geometric sphere packing. This \textit{a priori} stability prediction approach, or more specifically, structurally agnostic approach, uses experimentally tabulated values of ionic radii, such as the Shannon radii \cite{Shannon1976}. Such tools may be developed through geometric considerations to fit observed experimental trends in known materials. They are empirical devices which often work quite well within a particular material family, and with refinements may sometimes also be applied to similar compounds for which they were not initially developed. \cite{Bassen2024ToleranceSeries,Teraoka1998,Bartel2019,Sun2020StructuralDesign,Kieslich2015AnPerovskites,Sato2016ExtendingOPEN,CaiTheStructure}

\begin{figure}[ht]
  \centering
  \includegraphics[width = 330pt]{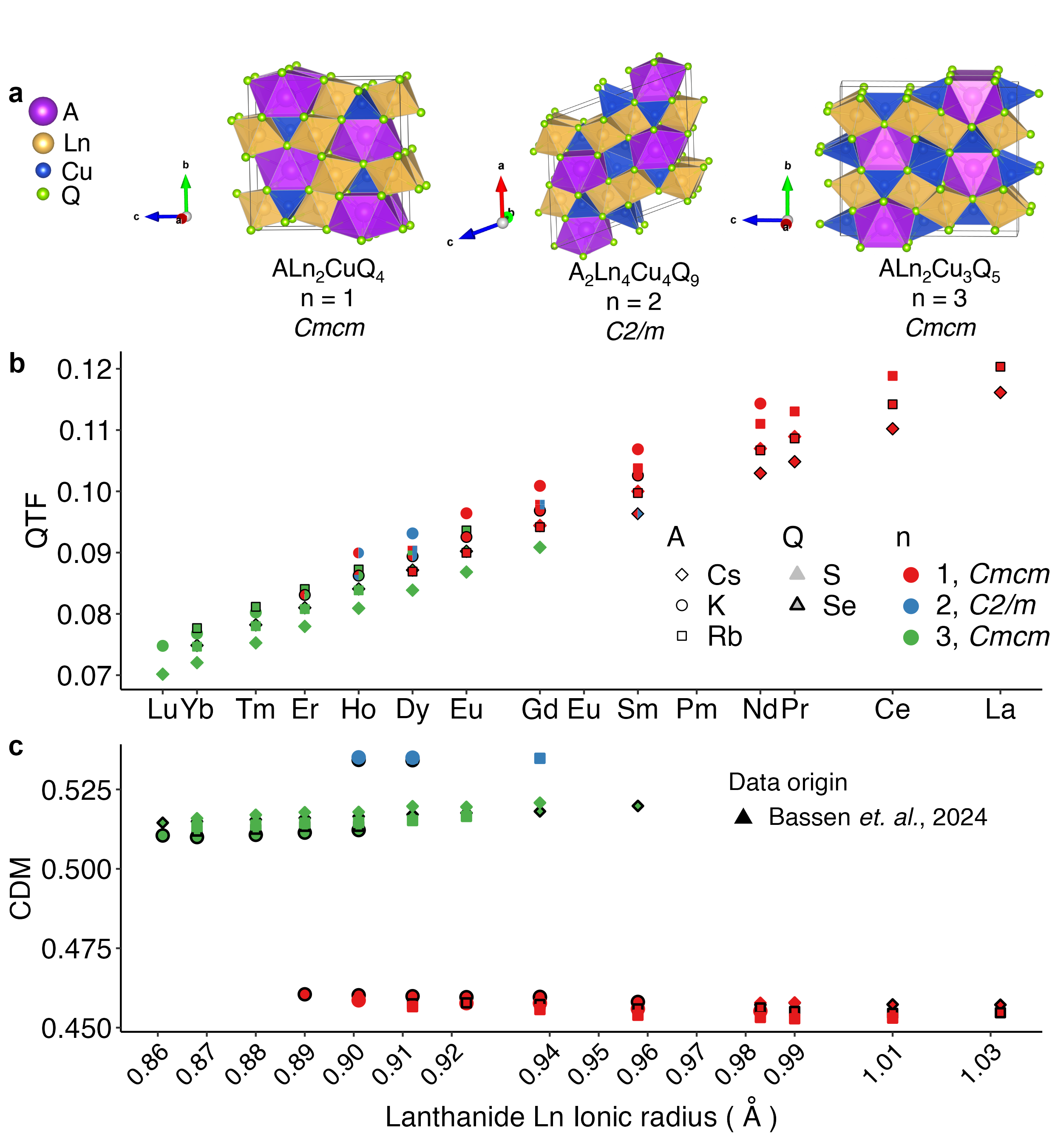}
  \caption{a) Structures n~=~1,~2 and 3 for the A\textsubscript{2}Ln\textsubscript{4}Cu\textsubscript{2n}Q\textsubscript{7+n} homologous series. b) Plot of tolerance factor (QTF) for homology prediction across series c) Cubic deviation metric CDM for structure comparison across series. Bolded points Q = Se; unbolded Q = S. Colors indicate different homologies and shapes indicate the different A-site elements. All structural information is discussed in this work.\protect\cite{Bassen2024ToleranceSeries}}
  \label{QTF}
\end{figure}

Homologous series are excellent platforms for describing structural evolution across variations of composition space. As a case study, we will take the quaternary homologous series A\textsubscript{2}Ln\textsubscript{4}Cu\textsubscript{2n}Q\textsubscript{7+n}, where interstitial Cu\textsubscript{2}Q groups are added upon increasing homology number (where A = alkali cation, Ln = lanthanide cation, Q = chalcogenide anion, for $n = 1, 2,  3$) \cite{Bassen2024ToleranceSeries}. The n~=~1 and ~3 homologies crystallize in the orthorhombic \textit{Cmcm} space gorup, and the n~=~2 homology takes the monoclinic \textit{C2/m} space group as shown in Figure~\ref{QTF}a. Using this series as a case study, we can compare the two primary approaches employed for understanding structural evolution, namely pre-synthesis methods such as tolerance factors and post-synthesis methods like CDM, to expound on the complementary nature between them. Two fundamental questions capture these approaches: 1) How does the homology number change as a function of composition? and 2) How does the unit cell change as a function of homology number? In previous work,\cite{Bassen2024ToleranceSeries} we sought to answer the first question and developed the QTF (Equation~\ref{eq:QTF}) to predict the synthesizability of homology formation (n~=1,~2,~3) within this series using atomic radii.

\begin{equation}
   \mathcal{QTF}=\frac{r_{\textit{medium cation}}-r_{\textit{smallest cation}}}{r_{\textit{largest cation}}+r_{\textit{anion}}}
    \label{eq:QTF}
\end{equation}  

Here, the large, medium, and small cations correspond to the alkali, lanthanide, and copper radii, respectively, and the anion corresponds to the radius of the chalcogenide. As shown in Figure~\ref{QTF}b, the QTF is successful in predicting regions of homology favorability, wherein the n~=~1 is favored at higher QTF and the n~=~3 at lower QTF, with a mixed-phase region in between where all homologies are favorable (note that some compositions, such as K-Ho-Cu-Se can make all three homologies). We observe that the homology formation can be described and predicted as a function of constituent atomic radii. However, the QTF, like all structural tolerance factors, is agnostic to changes in the unit cell's shape, treating the problem purely as that of sphere packing. Therefore, question 2 above cannot be answered using the QTF as it requires post-synthetic knowledge of the resulting compositions and crystal structures. 

Upon determination of the crystallographic structure, the unit cell parameters can be obtained and CDM calculated. Now, the homologies can be sorted cleanly and the structures can be more quantitatively understood. With the CDM's built-in normalization condition, it becomes clear how each homology takes on a unique structure. Figure~\ref{QTF}c illustrates the clustering of n~=~1,~2 and 3 homologies from each other. The n=1 homology has the lowest CDM $\approx0.46$, indicating higher unit cell symmetry than the other homologies. This lower CDM value can be understood by analyzing the local symmetry of the pentagonal slice of the bicapped trigonal prismatic alkali (Figure~\ref{QTF}a). Four of the edges are shared with the rare-earth octahedra, with one edge shared with the copper tetrahedra. However, in the case of the n~=~2 homology, the addition of Cu\textsubscript{2}Q replaces one of the side edge-sharing rare-earth octahedra with two edge-sharing copper tetrahedra, resulting in the loss of mirror symmetry around this central alkali and the overall lower symmetry to the monoclinic \textit{C2/m} space group. The cubic deviation metric for n~=~2 members is CDM $\approx.54$, a distortion of about $\approx+0.08$ from n~=~1.
 
 Finally, for the n~=~3, the further addition of interstitial Cu\textsubscript{2}Q replaces the remaining edge sharing rare earth octahedra along the $c$ axis, restoring mirror symmetry around the alkali, resulting in lower values relative to the n~=~2 of CDM $=.515$. However, the metric value is still larger than those of n~=~1 by about $0.055$. This discrepancy highlights the utility of the cubic deviation metric; although both n~=~1 and n~=~3 homologies take the orthorhombic \textit{Cmcm} space group, the two additional Cu\textsubscript{2}Q units in n~=~3 distorts the structure, such that its cubic deviation is more similar to n~=~2 than n~=~1. This approach to quantifying unit cell deviation allows for the disentangling of structural differences between homologies of shared space groups. Together with the QTF, CDM illustrates the complementary nature of analyzing structural evolution through both pre- and post-synthetic methods, whereby homologies can be predicted before synthesis and structurally compared afterward. 

\subsection{Structure-property relationships with multiple ways of describing a unit cell}\label{sec:3}
\begin{figure}[ht]
  \centering
  \includegraphics[width = 275pt]{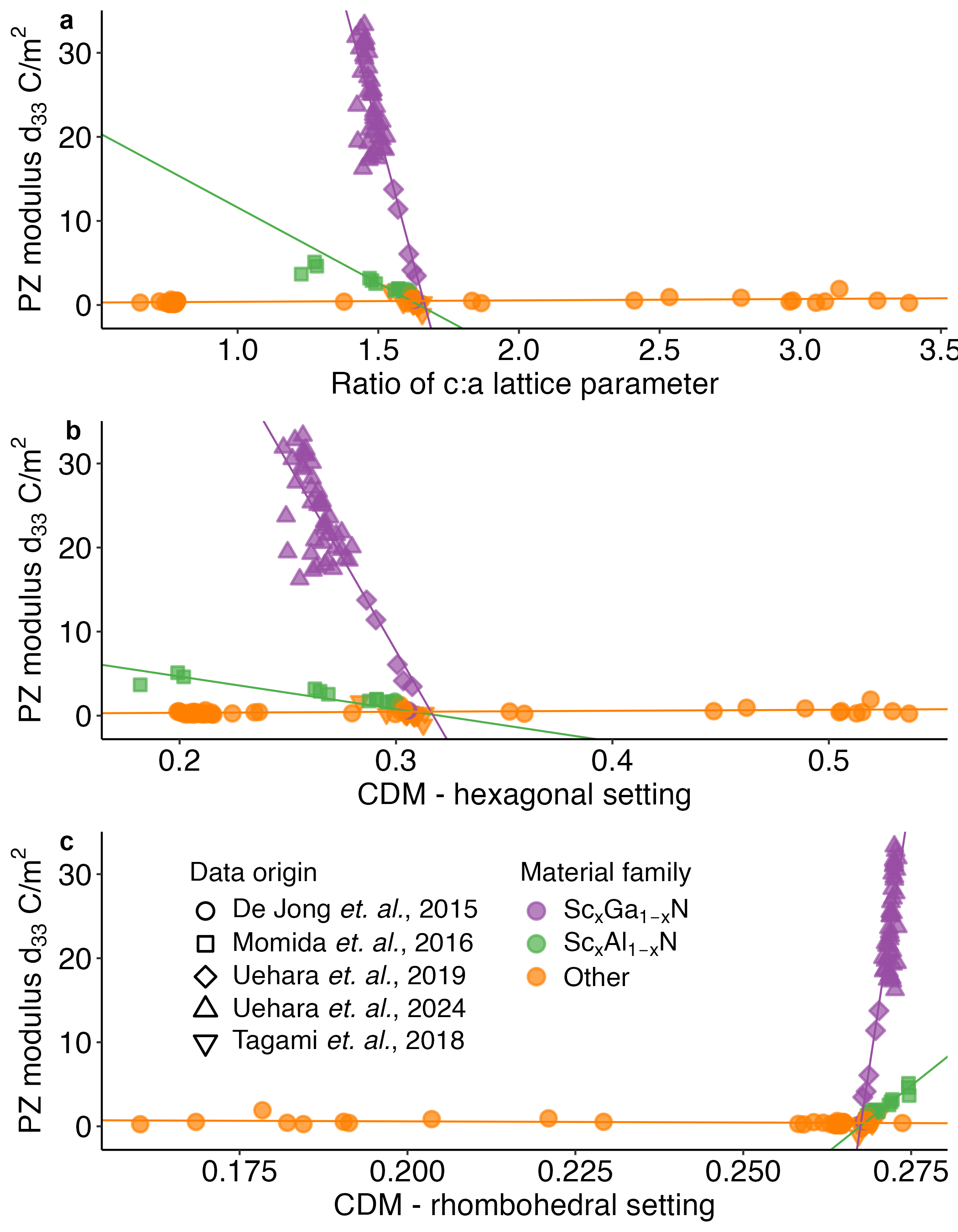}
  \caption{Piezoelectric modulus $d_{33}$ for spacegroup \textit{P63mc} materials as a function of a) $c:a$ ratio in the hexagonal setting (reported data), b) CDM in the hexagonal setting, and c) CDM in the rhombohedral setting (calculated). Lines of best fit indicate trends within materials families. Data is reproduced from \protect\cite{DeJong2015ASummary,Momida2016StrongCalculations,Uehara2019IncreaseMEMS,Uehara2024ExcellentConcentration,Tagami2018ToOguchi}.}
  \label{fig:Piezoca}
\end{figure}
One of the most obvious applications for a structural metric is in analysis of a property with a clear dependence on structure, such as piezomagnetism or piezoelectricity. For example, the Wurtzite-structure family of binary (AX) piezoelectrics has already been shown both experimentally\cite{Uehara2019IncreaseMEMS,Yazawa2021ReducedAl0.7Sc0.3N,Yasuoka2022TunableAnisotropy,Uehara2024ExcellentConcentration,Ota2025ImpactFilms} and computationally\cite{Jain2013Commentary:Innovation,DeJong2015ASummary,Momida2016StrongCalculations,Tagami2018ToOguchi} to correlate a key figure of merit, the piezoelectric modulus $d_{33}$ (C/m\textsuperscript{2}), with a structural parameter: namely, the ratio of unit cell length $c$ to width $a$, Figure~\ref{fig:Piezoca}a. As illustrated in \S \ref{sec:2} above, lattice parameter changes often have an origin in changes to the bonding in a material. In the Wurtzite crystal structure, these changes derive from alterations in the cation coordination and position along the $c$-axis.\cite{Uehara2019IncreaseMEMS,Tagami2018ToOguchi} The separation of cations and anions in turn drives a spontaneous polarization along the $c$-axis.\cite{Ota2025ImpactFilms} It is altogether unsurprising then that the $c:a$ ratio can correlate with $d_{33}$ and other parameters such as the minimum coercive field of ferroelectricity $E_c$ (MV/cm). For these reasons, $c:a$ optimization been the focus of several synthetic studies.\cite{Uehara2019IncreaseMEMS,Uehara2024ExcellentConcentration,Ota2025ImpactFilms}

The Wurtzite crystal structure, \textit{P63mc}, is hexagonal and, as noted previously, CDM is applicable in this setting with the caveat that the minimum possible value is nonzero. As Figure~\ref{fig:Piezoca}b illustrates, applying the cubic deviation metric to the Wurtzite family with the conventional hexagonal description reproduces the general trend seen with $c:a$ ratio in Figure~\ref{fig:Piezoca}a. This is expected behavior because with angles and $c$ constant, a decrease in $c:a$ ratio is the same as a decrease in deviation from a cube. Hence, a more cubic shape (lower CDM) correlates with higher $d_{33}$ within a particular material family. More interestingly, we can also apply CDM to the rhombohedral setting, as shown in Figure~\ref{fig:Piezoca}c, and again identify a trend, this time in a way which would be impossible with the conventional $c:a$ ratio test (since $c_\text{r}=a_\text{r}$ in this setting). Here we observe that in the alternate setting $d_{33}$ is \textit{lower} for more cubic materials. This can be understood through the relationship of the hexagonal to the rhombohedral setting (Figure~\ref{fig:HexRhomb}) - when the $c_\text{h}$ axis is changed, the rhombohedron inside is "squished" or "stretched", resulting in changes to the rhombohedral angles which in turn cause $\text{CDM}_\text{r}$ and $\text{CDM}_\text{h}$ to display different behavior (Figure~\ref{fig:MinRhombandHex}).

As discussed in detail in \S 2 above, CDM offers the benefit of more general applicability and automatic normalization, allowing one to quickly explore the connection between piezoelectric modulus and structure in additional materials families where $a$ and $c$ are not the only varying parameters. As a final demonstration of this, the materials with calculated piezoelectric modulus in the Materials Project database\cite{Jain2013Commentary:Innovation,DeJong2015ASummary} are plotted as a function of CDM in Figure~\ref{fig:PiezoMP}. The trend observable in the hexagonal lattices (Figure~\ref{fig:Piezoca}b-c) is not continued in other spacegroups, although the authors note significant errors in this dataset which also contains exclusively calculated, not experimental, data. The existing data suggest that the physics of the Wurtzite family specifically is what allows the increase in piezoelectric figure of merit with elongating $c$ axis. If these data accurately capture experimental trends, this reduces the parameters to target when synthesizing new piezoelectrics with different spacegroups. CDM also provides a common language for use by researchers working on different classes of piezoelectrics where the $c:a$ ratio may not always be a meaningful value, although as illustrated by Figure~\ref{fig:Piezoca}b-c, the physical behavior may manifest in different ways in different spacegroups.
\begin{figure}[H]
  \centering
  \includegraphics[width = 3in]{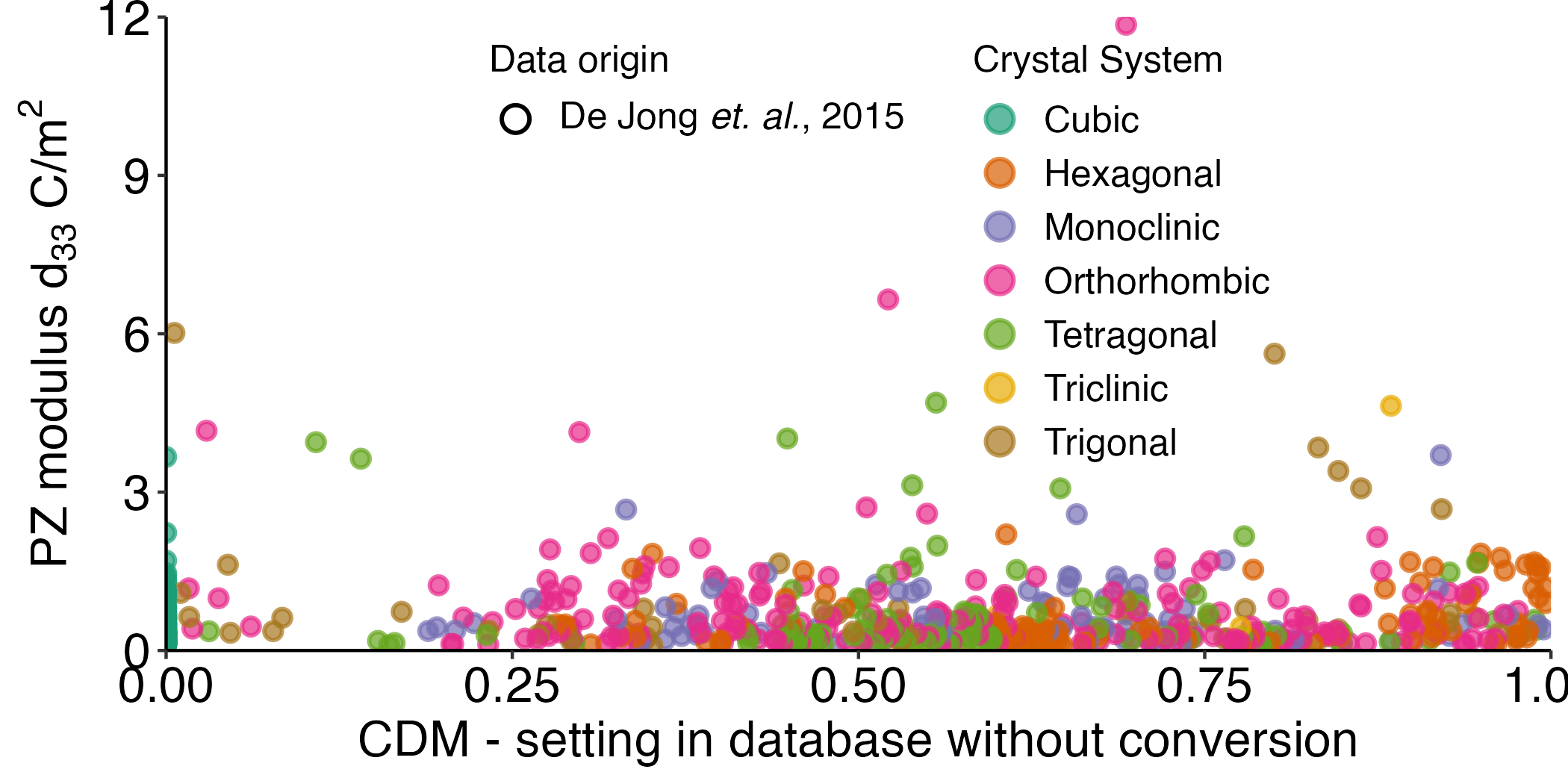}
  \caption{Piezoelectric modulus $d_{33}$ as a function of CDM for all piezoelectric materials in the Materials Project database, regardless of Laue class.\protect\cite{Jain2013Commentary:Innovation,DeJong2015ASummary}}
  \label{fig:PiezoMP}
\end{figure}

\subsection{Application to materials design}\label{sec:4}
As we have mentioned previously, the CDM may be used to interpret materials data after measurement and can be very useful as a tool to illustrate structure-property relationships. Consider, for example, the superconducting critical temperatures in the family of cuprate superconductors. The discovery of high-temperature superconductivity in La\textsubscript{2}CuO\textsubscript{4} \cite{Bednorz1986} led to an explosion of synthesis in the cuprate phase-space. It is well known that superconductivity in cuprates emerges from the essential CuO\textsubscript{2} planes within the structure\cite{zhang_effective_1988,bassen_real-space_2025}. Therefore, numerous interrelated perovskite-derived cuprates were found to superconduct in structurally and chemically analogous phase spaces containing these planes \cite{Park1995}. It was found that T\textsubscript{c} can vary considerably as a function of distances within or between CuO\textsubscript{2} planes \cite{Vanderah1992}. However, the relationship between cuprate unit cell and T\textsubscript{c}, agnostic to atomic position, remains a mystery. 
\begin{figure}[ht]
  \centering
  \includegraphics[width = 275pt]{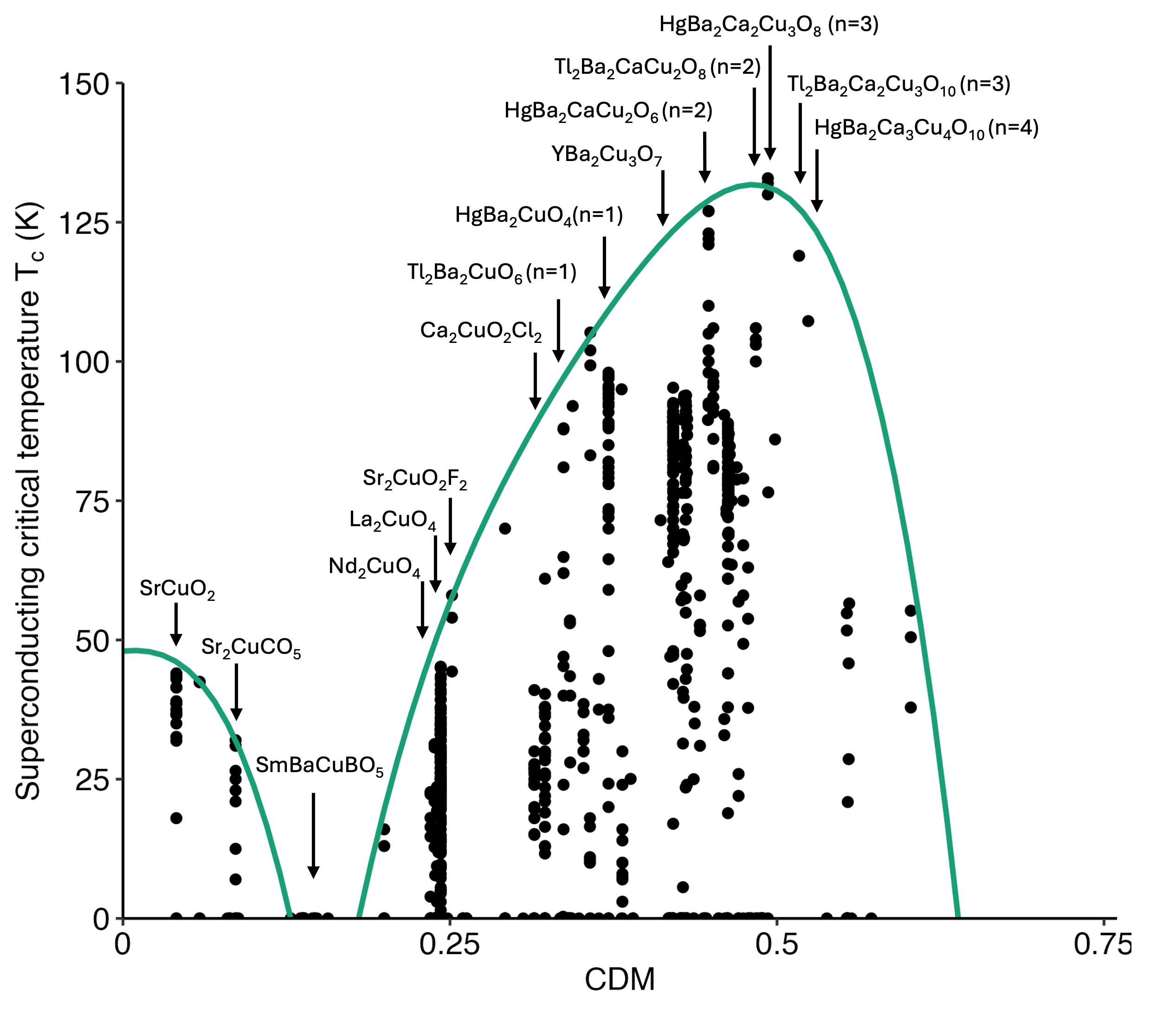}
  \caption{T\textsubscript{c} of cuprate superconductors in the 3DSC database\protect\cite{sommer_3dsc_2023,stanev_machine_2018} revealing two structural superconducting domes from polynomial fits of the maximum points. Materials with a T\textsubscript{c}~=~0 are not superconducting. The new materials synthesized in this work, SmBa$_{1-x}$K$_{x}$CuBO$_{5}$ (x=.05, 0.1, 0.2), lie between the two domes and show no evidence of superconductivity as described in the Appendix. Notable superconducting families are labeled with arrows above the plot. Vertical streaks indicate well-studied structural families that exhibit a range of T\textsubscript{c} values within the dataset.}
  \label{Cuprates}
\end{figure}

The cubic deviation metric allows for such a universal structural property comparison. Figure \ref{Cuprates} shows the relationship between T\textsubscript{c} and cubic deviation using all reported unit cell parameters of cuprates from the 3DSC, and thus SuperCon, database \cite{sommer_3dsc_2023,stanev_machine_2018}. Two superconducting domes emerge, represented by polynomial fits to the maximum T\textsubscript{c} points, which serve as visual guides. These lines indicate the expected maximum potential T\textsubscript{c} corresponding to a given CDM value. Structural families are labeled above the plot according to their idealized stoichiometries. The highly explored families within this dataset are shown as vertical streaking lines with a range of T\textsubscript{c}. 

The most cubic cuprate with CDM $\approx0.06$ is the infinite layer SrCuO\textsubscript{2} structure type, the parent compound of the cuprate family. Interestingly, between CDM $= 0.1-0.2$, there are no reported cuprate superconductors in the dataset. This may suggest that the unit cell geometry corresponding to CDM $= 0.1-0.2$ is not favorable for superconductivity among cuprates. We test this hypothesis through the synthesis and doping series of the underexplored noncentrosymmetric family LnBaCuBO\textsubscript{5} with CDM $= 0.14$, which  exists in the region between the superconducting domes. Specifically, we synthesize SmBa$_{1-x}$K$_{x}$CuBO$_{5}$ (x=0.05, 0.1, 0.2) and observe no evidence of superconductivity, as shown in Figure~\ref{fig:SmBaCuBO5_XRD}-\ref{fig:SmBaCuBO5_MPMS}. Synthesis information is provided in the Appendix. Additional syntheses within the LnBaCuBO\textsubscript{5} phase space are provided in the data repository. 

Evidence for superconductivity re-emerges at CDM $\approx0.2$. This marks the beginning of the second cuprate structural dome. The vertical streak at CDM $\approx 0.24$ corresponds to the highly explored family La\textsubscript{2}CuO\textsubscript{4}. Note that between CDM $= 0.238-0.314$ are the members of the A\textsubscript{n+1}B\textsubscript{n}X\textsubscript{3n+1} n=1 ruddlesden-popper derived cuprate materials, which include the T' electron doped family Nd\textsubscript{2}CuO\textsubscript{4} and the tetragonally elongated oxyhalide Ca\textsubscript{2}CuO\textsubscript{2}Cl\textsubscript{2}
The trend of increasing T\textsubscript{c} as a function of CDM in this second dome can be explained as the reduction of symmetry relative to the cube. Specifically, as the unit cell elongates from  La\textsubscript{2}CuO\textsubscript{4} to HgBa\textsubscript{2}Ca\textsubscript{2}Cu\textsubscript{3}O\textsubscript{8} at CDM $=0.493$, CDM increases proportionally, capturing the reduction of symmetry along the c axis. This plot indicates that the distortion of the unit cell, agnostic to connectivity of the atoms within it, captures the changes in dimensionality that favors higher T\textsubscript{c} in cuprate superconductors. Above CDM $=0.493$, further reduction in symmetry results in a decreased maximum T\textsubscript{c}. This is illustrated by analyzing the homologous families HgBa\textsubscript{2}Ca\textsubscript{n-1}Cu\textsubscript{n}O\textsubscript{2n+2} (n= 1, 2, 3, 4), and Tl\textsubscript{2}Ba\textsubscript{2}Ca\textsubscript{n-1}Cu\textsubscript{n}O\textsubscript{2n+4} (n= 1, 2, 3). Here, the homology number $n$ corresponds to the number of consecutive CuO\textsubscript{2} planes in the unit cell. For the Hg series, as $n$ increases and additional copper oxide layers are interstitially incorporated into the unit cell, T\textsubscript{c} increases up to $n = 3$, but decreases beyond this point, as seen in the $n = 4$ member. The precise location of the maximum in this plot could be refined with a more extensive dataset and the ongoing discovery of cuprate superconductors.  

The utility of the Cubic Deviation Metric (CDM) in guiding the design of novel cuprate superconductors is now evident. As generative machine learning enables the creation of vast numbers of hypothetical structures, a robust down-selection strategy is essential to prioritize candidate materials for experimental realization. We suggest that theoretical structures of cuprates with CDM $\approx0.4-0.5$ should be prioritized for this down-selection, where new and perhaps higher T\textsubscript{c} cuprate superconductors can be discovered. 

\section{Conclusion}
Herein we have constructed a metric CDM for the quantification of a unit cell's deviation from a perfect cube. The metric is unitless, volume-normalized, and applicable to all 7 crystal systems. It takes as inputs lattice parameters and provides a specific value CDM $=0$ when the unit cell is a cube, with predictable deviation as the unit cell becomes less cubic as individual lattice parameters vary continuously. Its identification of both a minimum and a maximum cubic deviation enables it to act as a quantitative descriptor of otherwise imprecise terms such as "pseudocubic" in describing real materials systems such as those covered in our four metric case studies.

Using the pseudobrookite materials family, we have shown the applicability of this metric to both continuous and discontinuous phase transitions. CDM varies approximately linearly with temperature and composition x in M\textsuperscript{3+}\textsubscript{3-x}Ti\textsuperscript{4+}\textsubscript{x}O\textsubscript{5} until a phase transition, in which case a change in slope and/or a jump in the metric is observed. CDM offers the benefit of agnosticism toward differing definitions of lattice parameters between space groups and simplifies the analysis of varying lattice parameters through a structural transition. This makes it easy to apply it simultaneously to a large number of materials, as illustrated in our second case study on the A\textsubscript{2}Ln\textsubscript{4}Cu\textsubscript{2n}Q\textsubscript{7+n} homologous series. Similar to a tolerance factor, CDM reveals trends with composition and groups different n-value homologies. However, unlike tolerance factors, the volume-normalized cubic deviation metric helps to illustrate the differences in structure in homologies even when they share the same space group, because of its incorporation of post-synthesis measurables (namely, lattice parameters).

Through our case study on Wurtzite binary piezoelectrics, we have also shown that the metric is applicable to multiple settings of the same lattice and can help demystify properties which initially appear to vary differently in different settings. This makes it readily applicable to problems of comparison between materials families where structure-property correlations known in one family or spacegroup may be obscured for others. We observe trends in piezoelectric modulus $d_{33}$ with CDM, reproducing trends typically illustrated by $c:a$ ratio in Wurtzite literature, and show that the physics of the Wurtzite piezoelectrics differs from that of other piezoelectric families. In our final case study, we illustrate similar concepts in the cuprate superconductors, to demonstrate how T\textsubscript{c} changes a function of CDM. We use data from the 3DSC database to show that there exist two structural domes describing T\textsubscript{c} as a function of CDM. Further, in between these domes, there is a region with no reported superconductivity in the dataset. We perform an experimental doping series of the SmBa$_{1-x}$K$_{x}$CuBO$_{5}$ family within this region to study whether superconductivity can emerge, but found no evidence. For tetragonally elongated superconductors, there exists a minimum $c$-axis symmetry wherein additional effects can work in tandem to increase T\textsubscript{c}, with diminishing returns beyond this point. This provides a framework for the design of new superconductors, as well as highlighting regions where otherwise promising candidates, such as SmBa$_{1-x}$K$_{x}$CuBO$_{5}$ (x=0.05, 0.1, 0.2), are not known to superconduct.

The cubic deviation metric CDM is an easy-to-use tool for quickly classifying unit cells and is likely to find applications in a wide array of materials problems. In particular, as more complete and error-free databases for materials are developed by and for AI applications, this tool may serve as a useful screening method or classifier for materials with known or suspected structure-property relationships. Since it enables structure analysis or prediction programs to identify nominally disparate crystal classes as occurring on a continuum, it may assist these programs in understanding the interrelated nature of unit cells, especially within doped or vacant structures - a weak spot for current methods. Open questions still remaining include its utility in mixed-phase systems such as solid solutions, locally disordered materials, and in \textit{in-situ} structural studies. We hope that this tool will prove useful for these and other applications and release with this work implementation of the metric in a variety of programming languages.
     
\appendix 
\fontfamily{phv}\selectfont
\normalfont
\section{Proof of metric}
\onehalfspacing
We shall prove three statements. First, we prove that $\mathcal{M}_\text{face}$ is zero only for square parallelograms (Figure~\ref{facediag}). Second, we show that $\mathcal{M}_\text{poly}$ is zero only when a parallelpiped $\mathcal{P}$ (Figure~\ref{fig:3Dangles}) is a cube, $(\mathcal{M}_\text{poly}=0) \iff (\text{parallelpiped $\mathcal{P}$ is a cube})$. Finally, we will prove that any parallelpiped scaled by any real factor will have equal $\mathcal{M}_\text{poly}$.

\subsection{Proof of Requirement 4, Part 1: squares (2D)}

\begin{theorem}
\label{thm:squares}
 For all parallelograms, Equation~\ref{eq:SquaresCase1} and Equation~\ref{eq:SquaresCase2} are true if and only if parallelogram MNOP is a square
\begin{equation}\label{eq:SquaresCase1}
\left|\ \frac{l_1}{\sqrt{l_1^2+l_2^2-2l_2l_1\cos{(\psi_{MNOP})}}}-\frac{1}{\sqrt{2}} \right|\ = 0 
\end{equation}

\begin{equation}\label{eq:SquaresCase2}
    \left|\ \frac{l_2}{\sqrt{l_1^2+l_2^2-2l_1l_2\cos{(\psi_{MNOP})}}}-\frac{1}{\sqrt{2}} \right|\ = 0 
\end{equation}

\end{theorem}
\vspace{5mm}
\begin{proof}
We will first prove the forward direction with a proof by cases, and then the backward direction. The first case is if parallelogram MNOP is a square then Equation~\ref{eq:SquaresCase1} is true. The second case is if parallelogram MNOP is a square then Equation~\ref{eq:SquaresCase2} is true.

\vspace{5mm}
\textbf{Forward direction proof Theorem~\ref{thm:squares}.}
If Equation~\ref{eq:SquaresCase1} is true and Equation~\ref{eq:SquaresCase2} is true then parallelogram MNOP is a square.
\begin{center}
    \includegraphics[scale=0.5]{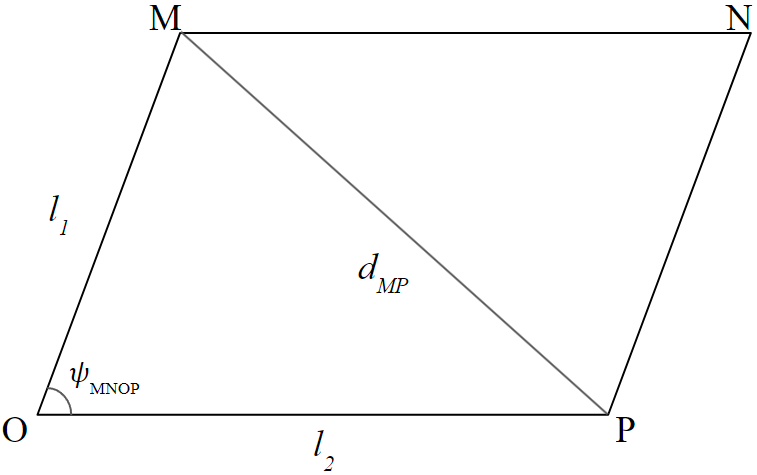}
\end{center}
\begin{align*}
    & \left|\ \frac{l_1}{\sqrt{l_1^2+l_2^2-2l_1l_2\cos{(\psi_{MNOP})}}}-\frac{1}{\sqrt{2}} \right|\ = 0   &&\text{(Given)}\\
    & \frac{l_1}{\sqrt{l_1^2+l_2^2-2l_1l_2\cos{(\psi_{MNOP})}}} = \frac{1}{\sqrt{2}}    &&\text{(Identity element (R,+))}\\
    & \frac{l_1*\sqrt{2}}{\sqrt{l_1^2+l_2^2-2l_1l_2\cos{(\psi_{MNOP})}}} = 1    &&\text{(Inverses (R,*))}
\end{align*}

 Using the Law of Cosines on $\triangle{MOP}$ we can see that  $d_{MP}= \sqrt{l_1^2+l_2^2-2l_1l_2\cos{(\psi_{MNOP})}}$, which is our denominator.
\begin{align*}
    & \frac{l_1*\sqrt{2}}{d_{MP}} = 1    &&\text{(substitution)}\\
    & l_1*\sqrt{2}  = d_{MP}    &&\text{(Inverses (R,+))}
\end{align*}

Using $\left|\ \frac{l_2}{\sqrt{l_1^2+l_2^2-2l_1l_2\cos{(\psi_{MNOP})}}}-\frac{1}{\sqrt{2}} \right|\ = 0$ the identical argument is made to find $l_2*\sqrt{2}  = d_{MP}$
\begin{align*}
    & l_1*\sqrt{2}  = d_{MP}    &&\text{ }\\
    & l_1*\sqrt{2}  = l_2*\sqrt{2}    &&\text{($l_2*\sqrt{2}  = d_{MP}$ substitution) }\\
    & l_1  = l_2    &&\text{(inverses (R,*)) }
\end{align*}
We now have all sides of the parallelogram MNOP congruent, and the hypotenuse of MOP equal to $l_1*\sqrt{2}$. since sides $l_1$,$l_2$ and hypotenuse $l_1*\sqrt{2}$ satisfy the Pythagorean equation, $\triangle $MOP must be a right triangle. Thus we have proven the forward direction.

\vspace{5mm}
 \textbf{Backward direction proof Theorem~\ref{thm:squares} - Case 1.} (parallelogram MNOP is a square)$\implies$ Equation~\ref{eq:SquaresCase1} is true.

\begin{center}
    \includegraphics[scale=0.33]{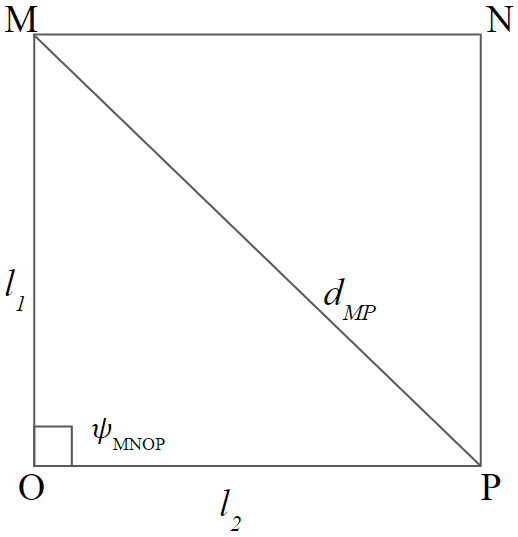}
\end{center}

Using the Law of Cosines on $\triangle MOP$ we have $d_{MP} = \sqrt{l_1^2+l_2^2-2l_1l_2\cos{(\psi_{MNOP})}}$ . Because $\triangle MOP$ is an isosceles right triangle triangle (angles \ang{45}, \ang{45}, and \ang{90}), we also know $d_{MP}=l_1*\sqrt{2}$.
 \begin{align*}
    &l_1*\sqrt{2} = d_{MP}  & &&\text{(from square)}\\
    &\frac{l_1*\sqrt{2}}{d_{MP}} = \frac{d_{MP}}{d_{MP}}  & &&\text{(important to note $d_{MP} \ne 0$)}\\
    &\frac{l_1}{d_{MP}} = \frac{1}{\sqrt{2}}  & &&\text{(Identity element)}\\
    &\left|\ \frac{l_1}{\sqrt{l_1^2+l_2^2-2l_1l_2\cos{(\psi_{MNOP})}}}-\frac{1}{\sqrt{2}} \right|\ = 0   & &&\text{($|0|=0$)}
\end{align*}

\vspace{5mm}
\textbf{Backward direction proof Theorem~\ref{thm:squares} - Case 2.} (parallelogram MNOP is a square)$\implies$ Equation~\ref{eq:SquaresCase2} is true.

It is trivial to show this using the exact same logic as above for the backward direction proof Theorem~\ref{thm:squares} Case 1. Thus we conclude:

(parallelogram MNOP is a square) $\implies $ \\ ($ \left|\ \frac{l_1}{\sqrt{l_1^2+l_2^2-2l_1l_2\cos{(\psi_{MNOP})}}}-\frac{1}{\sqrt{2}} \right|\ = 0$) $\land $ ($\left|\ \frac{l_2}{\sqrt{l_1^2+l_2^2-2l_1l_2\cos{(\psi_{MNOP})}}}-\frac{1}{\sqrt{2}} \right|\ = 0 $ )
\end{proof}

\vspace{5mm}
\begin{corollary}\label{cor:MSquare0iff}
 Let surface MNOP be a parallelogram. The Square Deviation Metric of Parallelogram MNOP $\mathcal{M}_{s_{MNOP}}=0$ if and only if parallelogram MNOP is a square.

 We define $\mathcal{M}_{s_{MNOP}}$ as the following.
 \begin{equation}
  \mathcal{M}_{s_{MNOP}}= \left|\ \frac{l_1}{\sqrt{l_1^2+l_2^2-2l_1l_2\cos{(\psi_{MNOP})}}}-\frac{1}{\sqrt{2}} \right|\ + \left|\ \frac{l_2}{\sqrt{l_1^2+l_2^2-2l_1l_2\cos{(\psi_{MNOP})}}}-\frac{1}{\sqrt{2}} \right|\ 
\end{equation}
\end{corollary}

\begin{proof}
From both backward direction proofs of Theorem~\ref{thm:squares} we have:
$( \left|\ \frac{l_1}{\sqrt{l_1^2+l_2^2-2l_1l_2\cos{(\psi_{MNOP})}}}-\frac{1}{\sqrt{2}} \right|\ = 0 )$ $\land $ $( \left|\ \frac{l_2}{\sqrt{l_1^2+l_2^2-2l_1l_2\cos{(\psi_{MNOP})}}}-\frac{1}{\sqrt{2}} \right|\ = 0 )$ $ \iff $ parallelogram MNOP is a square. 

\vspace{5mm}
Thus:
$( \left|\ \frac{l_1}{\sqrt{l_1^2+l_2^2-2l_1l_2\cos{(\psi_{MNOP})}}}-\frac{1}{\sqrt{2}} \right|\ + \left|\ \frac{l_2}{\sqrt{l_1^2+l_2^2-2l_1l_2\cos{(\psi_{MNOP})}}}-\frac{1}{\sqrt{2}} \right|\ = 0 )$ $ \iff $ parallelogram MNOP is a square.

\vspace{5mm}
The left hand side of this statement is $\mathcal{M}_{s_{MNOP}}$. Thus we arrive at
$( \mathcal{M}_{s_{MNOP}} = 0 )$ $ \iff $ parallelogram MNOP is a square.
\end{proof}

\subsection{Proof of Requirement 4, Part 2: cubes (3D)}
\begin{theorem}
\label{thm:cubes}
Let parallelpiped $\mathcal{P}$ have vertices M,N,O,P,Q,R,S,T. parallelpiped $\mathcal{P}$ has 3 surfaces that are squares joined at a common vertex (point O) if and only if parallelpiped $\mathcal{P}$ is a cube.
\begin{center}
    \includegraphics[scale=1]{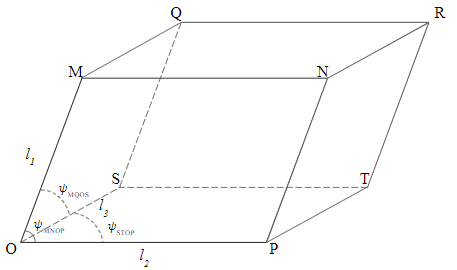}
\end{center}

\end{theorem}
\begin{proof} 
$\forall$ parallelpiped $\mathcal{P}$ (having vertices M,N,O,P,Q,R,S, and T); (parallelogram MNOP is a square) $\land$ (parallelogram STOP is a square) $\land$ (parallelogram MQOS is a square) $\iff$ (parallelpiped $\mathcal{P}$ is a cube).

\vspace{5mm}
\textbf{Forward direction of Theorem~\ref{thm:cubes}.} $\forall$ parallelpiped $\mathcal{P}$ (having vertices M,N,O,P,Q,R,S and T); [(parallelogram MNOP is a square) $\land$ (parallelogram STOP is a square) $\land$ (parallelogram MQOS is a square)] $\implies$ (parallelpiped $\mathcal{P}$ is a cube).

Let point O be the shared vertex of MNOP, OPST and MQOS. To prove parallelpiped $\mathcal{P}$ is a cube we shall first prove all faces of $\mathcal{P}$ are squares. By assumption, we already have MNOP, OPST and MQOS are squares with the same side length. Using the definition of a parallelpiped object we know each of the three faces MNOP, OPST and MQOS is congruent with its opposite face in $\mathcal{P}$. Thus all sides of $\mathcal{P}$ are squares. 

Since MNOP is a square $l_1=l_2$. We can now see that all edges of each face and of $\mathcal{P}$ are the same length.

Thus we have shown the three remaining sides of the parallelpiped $\mathcal{P}$ are squares, and all edges of the parallelpiped $\mathcal{P}$ have the same length $l_1$. Therefore by definition of a cube, parallelpiped $\mathcal{P}$ is a cube. 

\vspace{5mm}
\textbf{Backward direction of Theorem~\ref{thm:cubes}.} 
$\forall$ parallelpiped $\mathcal{P}$ (having vertices M,N,O,P,Q,R,S, and T); (parallelpiped $\mathcal{P}$ is a cube) $\implies$ [(parallelogram MNOP is a square) $\land$ (parallelogram STOP is a square) $\land$ (parallelogram MQOS is a square)] 

Since parallelpiped $\mathcal{P}$ is a cube, all sides are squares by the definition of a cube. Thus any three faces that share a common vertex will be squares. Thus the backwards direction is proved.

We conclude [(parallelogram MNOP is a square) $\land$ (parallelogram STOP is a square) $\land$ (parallelogram MQOS is a square)] $\implies$ (parallelpiped $\mathcal{P}$ is a cube).

\end{proof}

\begin{corollary}\label{cor:MCube0iff}
For all parallelpipeds, $\mathcal{M}_{C_{\mathcal{P}}} = 0$ (Cubic Deviation Metric of parallelpiped $\mathcal{P}$) if and only if $\mathcal{P}$ is a cube.

We define $\mathcal{M}_{C_{\mathcal{P}}}$ to be the following given the parallelpiped $\mathcal{P}$ (with vertices M,N,O,P,R,S,T) below. 
\begin{equation}\label{eq:MCP}
  \mathcal{M}_{C_{\mathcal{P}}}=  \mathcal{M}_{s_{MNOP}}+ \mathcal{M}_{s_{STOP}}+ \mathcal{M}_{s_{MQOS}}
\end{equation} 
\end{corollary}
\begin{proof} $\forall$ parallelpiped $\mathcal{P}$, $(\mathcal{M}_{C_{\mathcal{P}}}=0)$ $\iff \mathcal{P}$ is a cube.

From Theorem~\ref{thm:cubes} we have

$\forall$ parallelpiped $\mathcal{P}$, (parallelogram MNOP is a square) $\land$ (parallelogram STOP is a square) $\land$ (parallelogram MQOS is a square) $\iff$ (parallelpiped $\mathcal{P}$ is a cube).

From Corollary~\ref{cor:MSquare0iff} we know $(\mathcal{M}_{s_{MNOP}} = 0 )\iff$ parallelogram MNOP is a square, thus as follows:

$\forall$ parallelpiped $\mathcal{P}$, $[(\mathcal{M}_{s_{MNOP}}=0) \land (\mathcal{M}_{s_{STOP}}=0) \land (\mathcal{M}_{s_{MQOS}}=0)] \iff $ (parallelpiped $\mathcal{P}$ is a cube).

 $\forall$ parallelpiped $\mathcal{P}$, $[(\mathcal{M}_{s_{MNOP}} + \mathcal{M}_{s_{STOP}} + \mathcal{M}_{s_{MQOS}}=0) =0] \iff $ (parallelpiped $\mathcal{P}$ is a cube).

Now we see the metric $\mathcal{M}_{C_{\mathcal{P}}}$ on the left-hand side and substitute it to find the following

 $\forall$ parallelpiped H, $\mathcal{M}_{C_{H}} =0 \iff $ (parallelpiped $\mathcal{P}$ is a cube)
\end{proof}

The proof of Requirement 4 in higher dimensions follows in a similar fashion.

\subsection{Proof of Requirement 2: scalability}
\begin{theorem}\label{thm:MScaling}
    If parallelogram B the result of multiplying parallelogram A by a scalar, then both $\mathcal{M}_s$ are equal.
\end{theorem}
\begin{proof}
Given parallelogram A and B have positive area,
$\forall n, n \in \mathbb{R}, B=nA \implies \mathcal{M}_{s_B}=\mathcal{M}_{s_A}$

\begin{align*}
    \mathcal{M}_{s_A} &= \left|\ \frac{l_1}{\sqrt{l_1^2+l_2^2-2l_1l_2\cos{(\psi)}}}-\frac{1}{\sqrt{2}} \right|\ + \left|\ \frac{l_2}{\sqrt{l_1^2+l_2^2-2l_1l_2\cos{(\psi)}}}-\frac{1}{\sqrt{2}} \right|\  & &&\text{ }\\
    &= \frac{n}{n} (\left|\ \frac{l_1}{\sqrt{l_1^2+l_2^2-2l_1l_2\cos{(\psi)}}}-\frac{1}{\sqrt{2}} \right|\ + \left|\ \frac{l_2}{\sqrt{l_1^2+l_2^2-2l_1l_2\cos{(\psi)}}}-\frac{1}{\sqrt{2}} \right|\ ) &  &&\text{ }\\
     &= \left|\ \frac{nl_1}{\sqrt{(nl_1)^2+(nl_2)^2-2(nl_1)(nl_2)\cos{(\psi)}}}-\frac{1}{\sqrt{2}} \right|\ + \left|\ \frac{nl_2}{\sqrt{(nl_1)^2+(nl_2)^2-2(nl_1)(nl_2)\cos{(\psi)}}}-\frac{1}{\sqrt{2}} \right|\  & &&\text{ }\\
     \mathcal{M}_{s_A} &= \mathcal{M}_{s_B}  & &&\text{ } 
\end{align*}
Since parallelogram B has sides $nl_1,nl_2$ and $\angle \psi$ we can substitute $\mathcal{M}_{s_B}$.
\end{proof}
\begin{corollary}\label{cor:Scalability}
     If parallelpiped H the result of multiplying parallelpiped J by a scalar, then both $\mathcal{M}_C$ are equal.
\end{corollary}
\begin{proof}
Given parallelpiped H and J have positive area,
$\forall n, n \in \mathbb{R}, H=nJ \implies \mathcal{M}_{C_H}=\mathcal{M}_{C_J}$
\begin{align*}
    \mathcal{M}_{C_H} &= \mathcal{M}_{s_{h1}} + \mathcal{M}_{s_{h2}} + \mathcal{M}_{s_{h3}}  &&\text{(definition of $\mathcal{M}_{C_H}$)}\\
    &= \mathcal{M}_{s_{j1}} + \mathcal{M}_{s_{j2}} + \mathcal{M}_{s_{j3}}   &&\text{(Theorem~\ref{thm:MScaling})}\\
     \mathcal{M}_{C_H} &= \mathcal{M}_{C_J}   &&\text{(definition of $\mathcal{M}_{C_J}$)}
\end{align*}
\end{proof}

From Corollary~\ref{cor:MCube0iff} we have proven that $\mathcal{M}_{C}=0$ is equivalent to that parallelpiped being a cube. We now define $\mathcal{M}_\text{poly}$ as $\mathcal{M}_{C}$. Thus we have proven $\mathcal{M}_\text{poly}$ is only zero when the parallelpiped is a cube.

From Corollary~\ref{cor:Scalability} we have shown that $\mathcal{M}_C$, is unaffected by scaling, Thus we have shown any parallelpiped scaled by any real factor will have equal $\mathcal{M}_\text{poly}$. 

\subsection{Necessity of piecewise definition and requirement 3}
\begin{figure}[H]
\centering
    \includegraphics[scale=0.75]{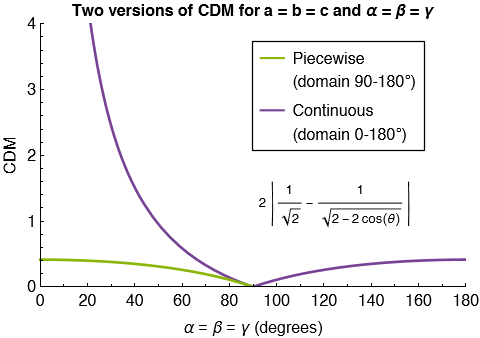}
    \caption{Plot showing metric value with restricted domain and without.}
    \label{fig:piecewise}
\end{figure}

We have already shown that the metric is only zero when taking angles between \ang{0} to \ang{180} however, due to the symmetry of the cosine graph when extended, there is another zero at \ang{270}. Since we have adjusted our graph to map the output of \ang{180} to \ang{270} onto the input of \ang{0} to \ang{90}, we have mapped this new zero at \ang{270} to our \ang{90}. The resulting piecewise function is continuous, meaning it is free from abrupt value changes and small changes in the input result in small changes in the output. The proof above holds for this new piecewise continuous function.

Since linearly scaling parallelpiped objects does not affect the angles, Corollaries~\ref{cor:MCube0iff} and \ref{cor:Scalability} hold. 

This piecewise definition creates a global maximum and adds symmetry between acute and obtuse angled parallelpiped objects. The five metric requirements laid out in the main text are preserved. CDM is only zero when the object has square faces, is independent of scalar multiplication of parallelpiped objects, depends only on the six lattice parameters, varies smoothly as the parallelpiped becomes less cubic, and matches the periodicity of the real lattice. 

\section{Additional pseudobrookite data}
\begin{figure}[H]
    \includegraphics[scale=0.8]{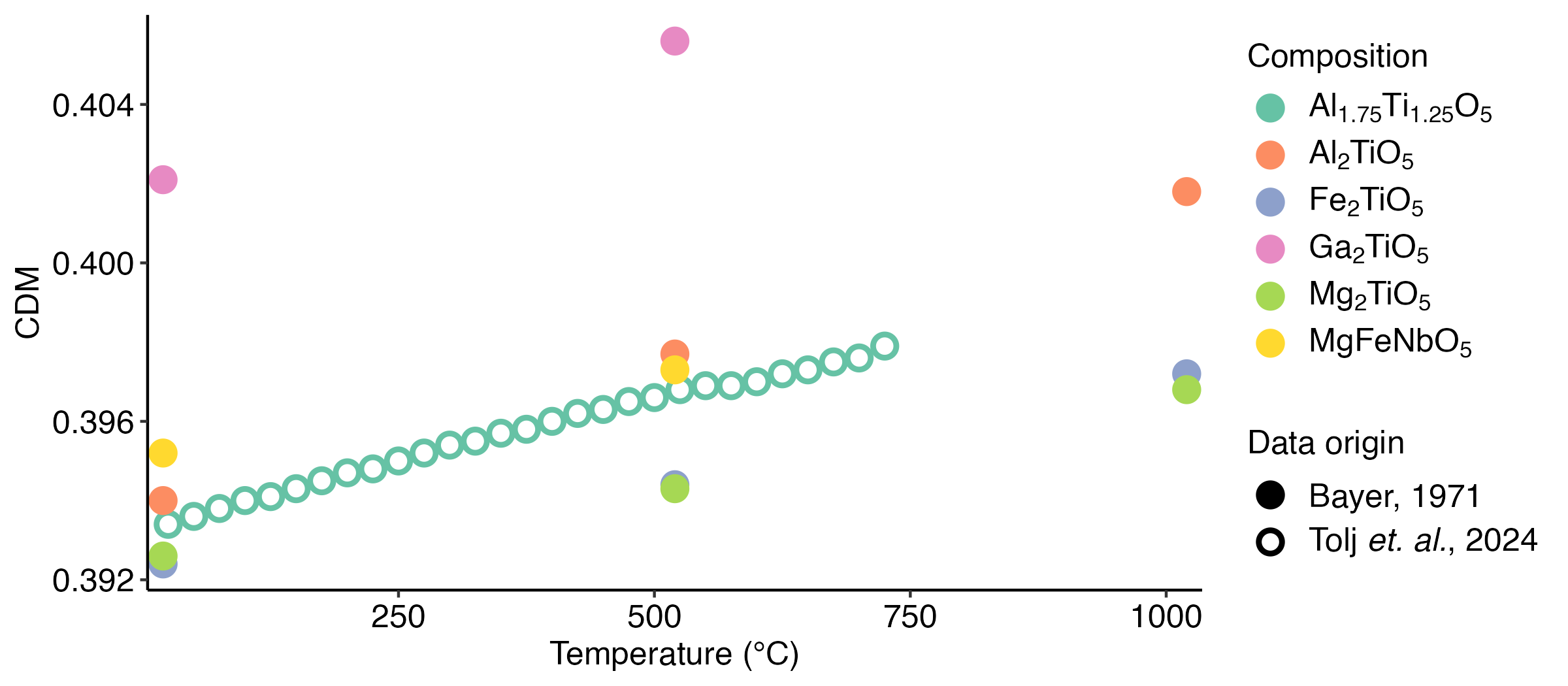}
    \caption{Pseudobrookite M\textsubscript{3-x}Ti\textsubscript{x}O\textsubscript{5} temperature variation.}
    \label{fig:PBextradata}
\end{figure}

\section{Rhombohedral and hexagonal lattice relationship}
\onehalfspacing
\begin{equation}
    \text{CDM}_{\text{hexagonal}}=\frac{2}{3}\left|\ \frac{a}{\sqrt{a^2+c^2}}-\frac{1}{\sqrt{2}} \right|\ +\frac{2}{3}\left|\ \frac{c}{\sqrt{a^2+c^2}}-\frac{1}{\sqrt{2}} \right|\ + \frac{2}{3}\left|\ \frac{1}{\sqrt{3}}-\frac{1}{\sqrt{2}} \right|\
    \label{eq:Mhex}
\end{equation}

\begin{equation}
    \text{CDM}_{\text{rhombohedral}}=2\left|\ \frac{1}{\sqrt{2-2\cos\alpha}}-\frac{1}{\sqrt{2}} \right|\
    \label{eq:Mrhomb}
\end{equation}

\begin{equation}
\begin{split}
    a_{\text{rhombohedral}}&=\frac{1}{3}\sqrt{3a^2_{\text{hexagonal}} + c^2_{\text{hexagonal}}} \\
    \\
    \alpha_{\text{rhombohedral}}&=\arccos\left(1-\frac{9a_{\text{hexagonal}}^2}{6a_{\text{hexagonal}}^2+2c_{\text{hexagonal}}^2}\right)\times\frac{180\degree}{\pi}
    \end{split}
    \label{eq:HexToRhomb}
\end{equation}

\begin{figure}[H]
\centering
    \includegraphics[scale=0.75]{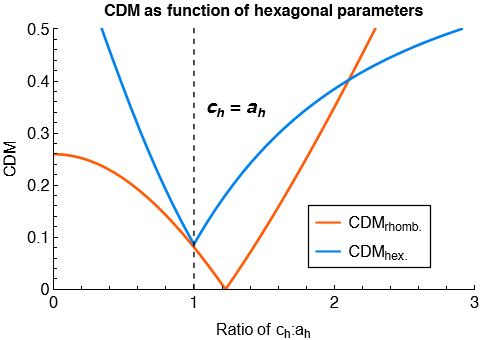}
    \caption{Plot showing the value of the metric for the rhombohedral and hexagonal setting as a function of the hexagonal lattice parameters. Observe that the value of $\text{CDM}_\text{rhomb.}$ is not minimized when $c_h=a_h$ even though $\text{CDM}_\text{hex.}$ is minimized.}
    \label{fig:MinRhombandHex}
\end{figure}

\section{New materials synthesized}\label{sec:QuentinData}
\subsection{Synthesis and characterization methods}
Materials in the family SmBa\textsubscript{1-x}K\textsubscript{x}CuBO\textsubscript{5} were synthesized from stoichiometric mixtures of ground H$_{3}$BO$_{3}$ (VWR, $99.5\%$), CuO (Thermoscientific, $99.995\%$), BaCO$_{3}$ (AlfaAesar, $99.6\%$), K$_{2}$CO$_{3}$ (VWR, $99\%$), Sm$_{2}$O$_{3}$ (NOAH, $99.9\%$). The ground powders were heated to $1070\,^{\circ}\mathrm{C}$ for 12 hours. The lattice parameters and phase purity of the materials were determined using room temperature powder x-ray diffraction (XRD) with Cu K$\alpha$ radiation. Susceptibility was measured using 50 oe applied field in a Quantum Design Magnetic Properties Measurement System (MPMS). No superconducting transition was observed for any of the materials synthesized here.

XRD and magnetization data for the materials in Table~\ref{table:QuentinData} are presented in the sections below.  The cubic deviation metrics for these materials are presented in Figure~\ref{Cuprates}. Additional data may be found in our data repository.

\begin{table}[hbt]
\begin{tabular}{llccccl}
\textbf{Composition}        & \textbf{a}       & \textbf{b}       & \textbf{c}       & \textbf{$\alpha,\beta,\gamma$} & \textbf{CDM} \\SmBa$_{0.95}$K$_{0.05}$CuBO$_5$           & 5.49667  & 5.49667          & 7.40250         & \ang{90}                          & 0.1378     \\
SmBa$_{0.9}$K$_{0.1}$CuBO$_5$             & 5.49375 & 5.49375          & 7.39864          & \ang{90}                          & 0.1378     \\
SmBa$_{0.8}$K$_{0.2}$CuBO$_5$             & 5.49209         & 5.49209         & 7.40083          & \ang{90}                          & 0.1380     \\ 
\end{tabular}
\caption{New materials synthesized in this work}
\label{table:QuentinData}
\end{table}

\subsection{SmBa\textsubscript{x}K\textsubscript{1-x}CuBO\textsubscript{5} results}
\begin{figure}[H]
\includegraphics[scale=0.8]{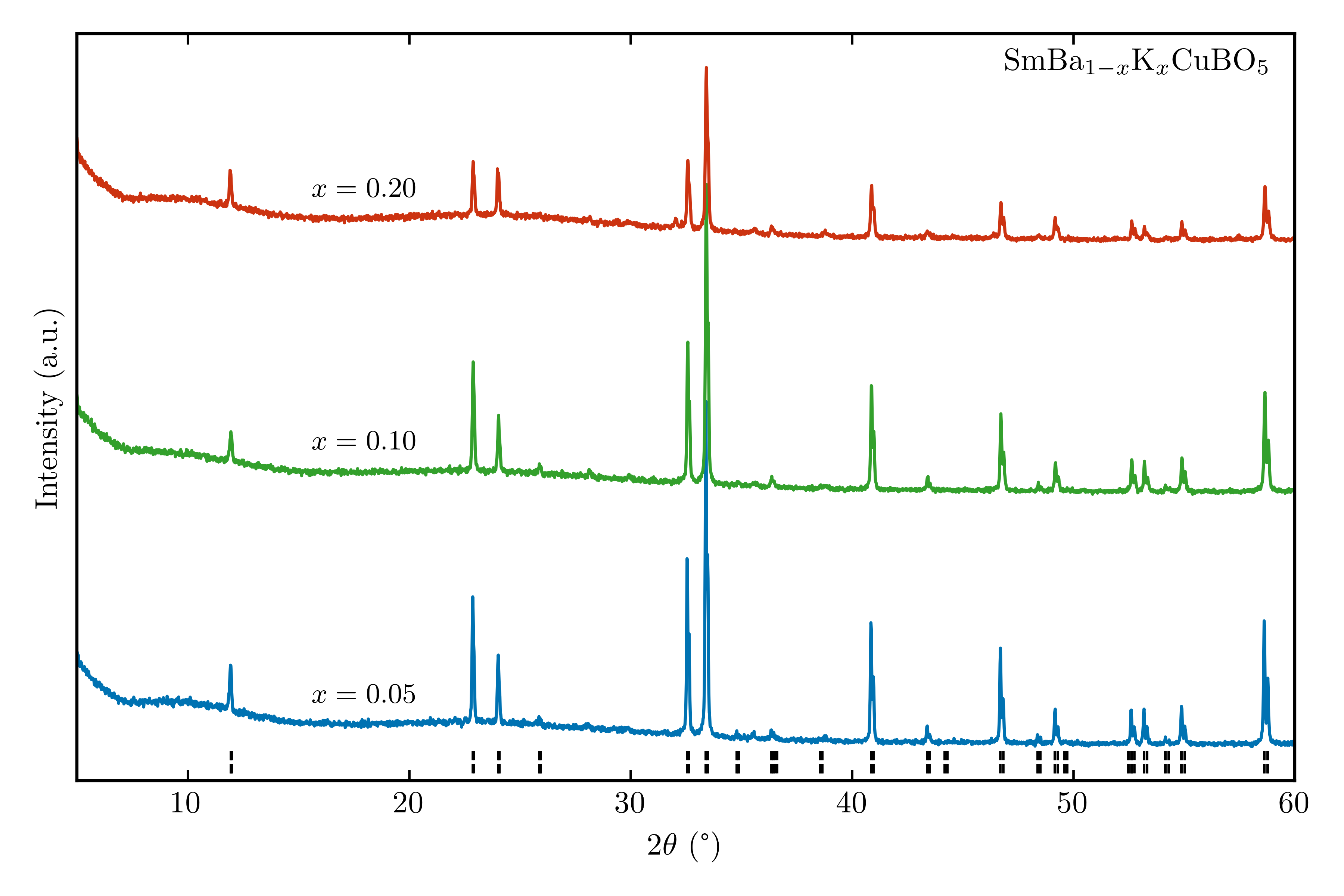}
    \caption{XRD for SmBa$_{1-x}$K$_{x}$CuBO$_{5}$ across the doping series x=0.05, 0.1, 0.2. Expected peak positions of SmBaCuBO$_{5}$ are denoted by vertical dashed lines.}
    \label{fig:SmBaCuBO5_XRD}
\end{figure}

\begin{figure}[H]
    \includegraphics[scale=0.65]{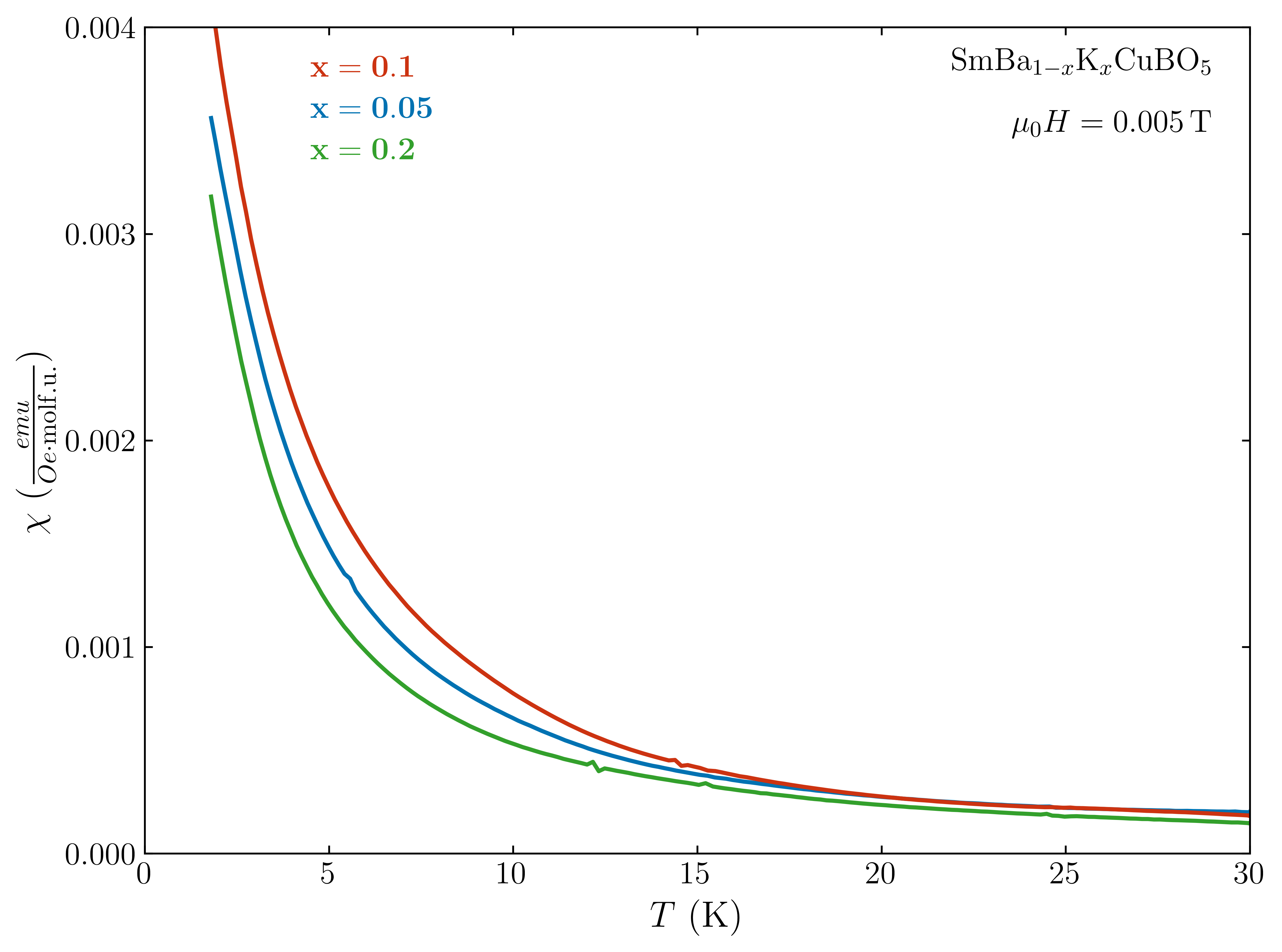}
    \caption{Susceptibility as a function of temperature for SmBa$_{1-x}$K$_{x}$CuBO$_{5}$ across the doping series x=0.05, 0.1, 0.2 with a H= 50 Oe applied field from 1.8-30K.}
    \label{fig:SmBaCuBO5_MPMS}
\end{figure}
\begin{acknowledgements}
We thank Ryan Baumbach for discussing an early case study for this manuscript. QS acknowledges the PARADIM REU program. We thank the Johns Hopkins University for support of this work.  
\end{acknowledgements}

\begin{funding}
QS acknowledges the PARADIM REU program funded by the NSF under award DMR‑2039380. DT acknowledges funding through the William H. Miller III Postdoctoral Fellowship in the Department of Physics and Astronomy, Johns Hopkins University. This work was funded by the U.S. Department of Energy (DOE), Office of Science (SC), National Quantum Information Science Research Centers, Co-Design Center for Quantum Advantage (C2QA) under contract number DE-SC0012704. The MPMS3 system used for magnetic characterization was funded by the National Science Foundation, Division of Materials Research, Major Research Instrumentation Program, under Award No. 1828490.
\end{funding}

\ConflictsOfInterest{The authors declare no conflicts of interest.
}

\DataAvailability{In addition to the Appendix, we release data and code here: \url{https://github.com/sbernierjhu/CDMdata}
}

\bibliography{references} 

\end{document}